\documentclass[a4paper]{article}

\usepackage{fullpage}

  \usepackage[utf8]{inputenc}
  \usepackage{amsmath,amsfonts,amsthm,amssymb}
  \usepackage{microtype}
  \usepackage{tikz}
  \usepackage{tikz-qtree}
  \usetikzlibrary{decorations.pathmorphing}
  \usetikzlibrary{decorations.markings}
  \usepackage{xspace}
  \usepackage{mathtools}
  \usepackage{enumitem}
  \usepackage{authblk}
  \usepackage{hyperref}
  \usepackage[capitalise]{cleveref}
  \usepackage{subcaption}

  \usetikzlibrary{shapes,calc,arrows,positioning,decorations.markings,trees}
  \tikzset{
    treenode/.style = {align=center, inner sep=2pt, text centered,
      font=\sffamily},
    arn_r/.style = {treenode, circle, black, fill=black,font=\sffamily\bfseries, draw=black,
      text width=0.2em},
      arn_p/.style = {treenode, circle, line width=0.7mm, blue, fill=black, draw=blue, text width=0.4em},
      arn_p1/.style = {treenode, circle, line width=0.7mm, blue, fill=blue, draw=blue, text width=0.4em},
      arn_pat/.style = {treenode, circle , blue, fill=blue, draw=blue, text width=0.3em},
      arn_t/.style = {treenode, circle, black, thick, double, font=\sffamily\bfseries, draw=black,
      text width=0.2em},
    every edge/.append style={anchor=south,auto=falseanchor=south,auto=false,font=3.5 em},
  }
  \tikzset{
      pattern/.style={postaction={decorate},
          decoration={markings,mark=at position .55 with {\draw[thin,fill, blue] circle (5pt);}}}
  }

\def\dd{\mathinner{.\,.}}
\newcommand{\cO}{\mathcal{O}}
\newcommand{\Oh}{\cO}
\newcommand{\stt}{\mathcal{T}(T)}
\newcommand{\RMQ}{\textsf{RMQ}\xspace}
\newcommand{\gen}{\textsf{g}}
\newcommand{\val}{\textsf{val}}
\newcommand{\Tr}{\mathsf{PT}}
\newcommand{\BTr}{\mathsf{BT}}

\newcommand{\Ohtilde}{\tilde{\cO}}
\newcommand{\cOtilde}{\Ohtilde}

\newcommand{\occ}{|output|}
\newcommand{\D}{\mathcal{D}}
\newcommand{\Ru}{\mathcal{R}}
\newcommand{\Dec}{\textsc{Exists}}
\newcommand{\R}{\textsc{Report}}
\newcommand{\RD}{\textsc{ReportDistinct}}
\newcommand{\C}{\textsc{Count}}
\newcommand{\CD}{\textsc{CountDistinct}}
\newcommand{\id}{\textsf{id}}
\newcommand{\per}{\textsf{per}}
\newcommand{\Lab}{\mathcal{L}}
\newcommand{\LCA}{\textsf{LCA}\xspace}
\newcommand{\Z}{\mathbb{Z}}

 \newcommand{\defproblem}[3]{
  \vspace{2mm}
\noindent\fbox{
  \begin{minipage}{0.96\textwidth}
  #1\\
  {\bf{Input:}} #2  \\
  {\bf{Query:}} #3
  \end{minipage}
  }
  \vspace{2mm}
}

\DeclarePairedDelimiter{\floor}{\lfloor}{\rfloor}
\DeclarePairedDelimiter{\ceil}{\lceil}{\rceil}

\newtheorem{theorem}{Theorem}[section]
\newtheorem{fact}[theorem]{Fact}
\newtheorem{lemma}[theorem]{Lemma}
\newtheorem{observation}[theorem]{Observation}
\newtheorem{proposition}[theorem]{Proposition}

\newtheorem{claim}[theorem]{Claim}

\theoremstyle{remark}
\newtheorem{definition}[theorem]{Definition}

\newtheorem{example}[theorem]{Example}
\newtheorem{conjecture}[theorem]{Conjecture}

\begin{document}

\title{Internal Dictionary Matching}

\author[1]{Panagiotis Charalampopoulos}
\author[2,3]{Tomasz Kociumaka}
\author[1]{Manal Mohamed}
\author[2]{Jakub~Radoszewski}
\author[2]{Wojciech Rytter}
\author[2]{Tomasz Wale\'n}

\affil[1]{Department of Informatics, King's College London, London, UK\\\texttt{[panagiotis.charalampopoulos,manal.mohamed]@kcl.ac.uk}}
\affil[2]{Institute of Informatics, University of Warsaw, Warsaw, Poland\\\texttt{[kociumaka,jrad,rytter,walen]@mimuw.edu.pl}}
\affil[3]{Department of Computer Science, Bar-Ilan University, Ramat Gan, Israel}

\date{\vspace{-.5cm}}

\maketitle              
\begin{abstract}
We introduce data structures answering queries concerning
the occurrences of patterns from a given dictionary $\D$ in
\emph{fragments} of a given string $T$ of length $n$. The dictionary is \emph{internal}
in the sense that each pattern in $\D$ is given as a fragment of $T$.
This way, $\D$ takes space proportional to the number of patterns $d=|\D|$
rather than their total length, which could be $\Theta(n\cdot d)$.

In particular, we consider the following types of queries:
reporting and counting \emph{all} occurrences of patterns from $\D$ in a
fragment $T[i \dd j]$ and reporting \emph{distinct} patterns
from $\D$ that occur in $T[i \dd j]$. 
We show how to construct, in $\cO((n+d) \log^{\cO(1)} n)$ time, a data structure that answers each of these
queries in time $\cO(\log^{\cO(1)} n+\occ)$.

The case of counting patterns is much more involved and needs a combination
of a locally consistent parsing with orthogonal range searching.
Reporting distinct patterns, on the other hand, uses the structure of maximal repetitions in strings.
Finally, we provide tight---up to subpolynomial factors---upper and lower bounds
for the case of a dynamic dictionary.
\end{abstract}

\section{Introduction}

In the problem of dictionary matching, which has been studied for more than forty years, we are given a dictionary $\D$, consisting of $d$ patterns, and the goal is to preprocess $\D$ so that presented with a text $T$ we are able to efficiently compute the occurrences of the patterns from $\D$ in $T$. 
The Aho--Corasick automaton preprocesses the dictionary in linear time with respect to its total length and then processes $T$ in time $\cO(|T|+\occ)$~\cite{DBLP:journals/cacm/AhoC75}.
Compressed indexes for dictionary matching~\cite{DBLP:conf/soda/ChanHLS05}, as well as indexes for approximate dictionary matching~\cite{DBLP:conf/stoc/ColeGL04} have been studied.
Dynamic dictionary matching in its more general version consists in the problem where a dynamic dictionary is maintained, text strings are presented as input and for each such text all the occurrences of patterns from the dictionary in the text have to be reported; see~\cite{DBLP:journals/jcss/AmirFGBP94,DBLP:journals/iandc/AmirFIPS95}.

Internal queries in texts have received much attention in recent years.
Among them, the \emph{Internal Pattern Matching} (IPM) problem consists in preprocessing a text $T$ of length $n$ so that we can efficiently compute the occurrences of a substring of $T$ in another substring of $T$.
A nearly-linear sized data structure that allows for sublogarithmic-time IPM queries was presented in~\cite{DBLP:journals/tcs/KellerKFL14}, while a linear sized data structure allowing for constant-time IPM queries in the case that the ratio between the lengths of the two substrings is constant was presented in~\cite{DBLP:conf/soda/KociumakaRRW15}.
Other types of internal queries include computing the longest common prefix of two substrings of $T$, computing the periods of a substring of $T$, etc.
We refer the interested reader to~\cite{tomeksthesis}, which contains an overview of the literature.

We introduce the problem of \emph{Internal Dictionary Matching} (IDM) that consists in answering the following types of queries for an internal dictionary $\D$ consisting of substrings of text $T$: given $(i,j)$, report/count all occurrences of patterns from $\D$ in $T[i \dd j]$ and report the distinct patterns from $\D$ that occur in $T[i \dd j]$.

Some interesting internal dictionaries $\D$ are the ones comprising of palindromic, square, or non-primitive substrings of $T$.
In each of these three cases, the total length of patterns might be quadratic, but the internal dictionary is of linear size 
and can be constructed in $\Oh(n)$ time~\cite{DBLP:journals/ipl/GroultPR10,DBLP:journals/tcs/CrochemoreIKRRW14,DBLP:conf/cpm/BannaiIK17}.
Our data structure provides a general framework for solving problems related to the internal structure of the string.
The case of palindromes has already been studied in  \cite{DBLP:conf/spire/RubinchikS17},
where authors proposed a data structure of size $\cO(n\log n)$ that returns the number of all distinct palindromes in $T[i\dd j]$ in  $\cO(\log n)$  time. 

Let us formally define the problem and the types of queries that we consider.

  \vspace{2mm}
\noindent\fbox{
  \begin{minipage}{0.96\textwidth}
  \textsc{Internal Dictionary Matching} \\
  {\bf{Input:}} A text $T$ of length $n$ and a dictionary $\D$ consisting of $d$ patterns, each given as a substring $T[a \dd b]$ of $T$.  \\
  {\bf{Queries:}} \\
  $\Dec(i,j)$: Decide whether at least one pattern $P \in \D$ occurs in $T[i \dd j]$. \\
  $\R(i,j)$: Report all occurrences of all the patterns of $\D$ in $T[i \dd j]$. \\
  $\RD(i,j)$: Report all patterns $P \in \D$ that occur in $T[i \dd j]$. \\
  $\C(i,j)$: Count the number of all occurrences of all the patterns of $\D$ in $T[i \dd j]$.
  \end{minipage}
  }
  \vspace{2mm}

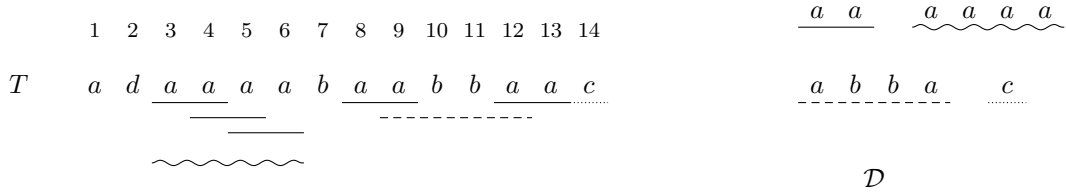
\begin{figure}[htpb]
\begin{center}
\begin{tikzpicture}
  \draw (-1,0) node[above] {$T$};
  \foreach \x/\n/\c in {0/1/a,0.5/2/d,1/3/a,1.5/4/a,2/5/a,2.5/6/a,3/7/b,3.5/8/a,4/9/a,4.5/10/b,5/11/b,5.5/12/a,6/13/a,6.5/14/c}{
    \draw (\x,0) node[above] {$\c$};
    \draw (\x,0.75) node[above] {\footnotesize \n};
  }
  \draw (0.75,0) -- (1.75,0);
  \draw[xshift=0.5cm] (0.75,-0.2) -- (1.75,-0.2);
  \draw[xshift=1cm] (0.75,-0.4) -- (1.75,-0.4);
  \draw[xshift=2.5cm] (0.75,0) -- (1.75,0);
  \draw[xshift=4.5cm] (0.75,0) -- (1.75,0);
  \draw[densely dotted] (6.25,0) -- (6.75,0);
  \draw[densely dashed] (3.75,-0.2) -- (5.75,-0.2);
  \draw[decorate, decoration={snake,amplitude=1pt}] (0.75,-0.8) -- (2.75,-0.8);

  \begin{scope}[xshift=9.5cm,yshift=1cm]
    \begin{scope}
      \foreach \x/\c in {0/a,0.5/a}{
        \draw (\x,0) node[above] {$\c$};
      }
      \draw (-0.25,0) -- (0.75,0);
    \end{scope}
    \begin{scope}[xshift=1.5cm,yshift=0cm]
      \foreach \x/\c in {0/a,0.5/a,1/a,1.5/a}{
        \draw (\x,0) node[above] {$\c$};
      }
      \draw[decorate, decoration={snake,amplitude=1pt}] (-0.25,0) -- (1.75,0);
    \end{scope}
    \begin{scope}[yshift=-1cm]
      \foreach \x/\c in {0/a,0.5/b,1/b,1.5/a}{
        \draw (\x,0) node[above] {$\c$};
      }
      \draw[densely dashed] (-0.25,0) -- (1.75,0);
    \end{scope}
    \begin{scope}[xshift=2.5cm,yshift=-1cm]
      \foreach \x/\c in {0/c}{
        \draw (\x,0) node[above] {$\c$};
      }
      \draw[densely dotted] (-0.25,0) -- (0.25,0);
    \end{scope}
    \draw (0.75,-2) node {$\D$};
  \end{scope}
\end{tikzpicture}
\end{center}
\caption{Occurrences of patterns from the dictionary $\D$ in the text $T$.}
\label{fig:main}
\end{figure}

\begin{example}\label{ex:1}
  Let us consider the dictionary $\D=\{aa,aaaa,abba,c\}$ and the text $T=adaaaabaabbaac$; see Fig.~\ref{fig:main}.
  We then have:
  \begin{align*}
    \Dec(2,12)&=\mathbf{true}\\
    \R(2,12)&=\{(aa,3),(aaaa,3),(aa,4),(aa,5),(aa,8),(abba,9)\}\\
    \C(2,12)&=6\\
    \RD(2,12)&=\{aa,aaaa,abba\}\\
    \Dec(1,3)&=\mathbf{false}
  \end{align*}
\end{example}

Let us consider $\R(i,j)$ queries. One could answer them in time $\cO(j-i+\occ)$ by running $T[i \dd j]$ over the Aho--Corasick automaton of $\D$~\cite{DBLP:journals/cacm/AhoC75} or in time $\cOtilde(d+\occ)$ \footnote{The $\tilde{\cO}(\cdot)$ notation suppresses $\log^{\cO(1)} n$ factors.} by performing internal pattern matching \cite{DBLP:conf/soda/KociumakaRRW15} for each element of $\D$ individually. None of these approaches is satisfactory as they can require $\Omega(n)$ time in the worst case.

\paragraph{\bf Our results.}
A natural problem would be to consider a dynamic dictionary, in the sense that one would perform interleaved IDM queries and updates to $\D$ (insertions/deletions of patterns).
We show a conditional lower bound for this problem. In particular, we show that the product of the time to process an update and the time to answer whether any pattern from $\D$ occurs in $T[i \dd j]$ cannot be $\cO(n^{1-\epsilon})$ for any constant $\epsilon>0$, unless the Online Boolean Matrix-Vector Multiplication conjecture~\cite{DBLP:conf/stoc/HenzingerKNS15} is false.
Interestingly, in our lower bound construction we only add single-letter patterns to an initially empty dictionary.

We thus focus on the case of a static dictionary, as it was defined above.
We propose an $\cOtilde(n+d)$-sized data structure, which can be built in time $\cOtilde(n+d)$ and answers all IDM queries in time $\cOtilde(1+\occ)$. The exact complexities are shown in Table~\ref{tab}.

\renewcommand*{\arraystretch}{1.35}  
\begin{table}[ht]
  \begin{center}
  \begin{tabular} {|c|c|c|c|}
    \hline
  Query & Preprocessing time & Space & Query time \\ \hline
  $\Dec(i,j)$ & $\cO(n+d)$ & $\cO(n)$ & $\cO(1)$ \\
  $\R(i,j)$ & $\cO(n+d)$ & $\cO(n+d)$ & $\cO(1+\occ)$ \\
  $\RD(i,j)$ & $\cO(n\log n+d)$ & $\cO(n+d)$ & $\cO(\log n+\occ)$ \\
  $\C(i,j)$ & $\cO(\frac{n\log n}{\log \log n} +d\log^{3/2} n)$ & $\cO(n+d\log n)$ & $\cO(\frac{\log^2n}{\log\log n})$ \\ \hline
  \end{tabular}
\end{center}
\caption{Our results.}\label{tab}
\end{table}

By building upon our solutions for static dictionaries, we provide algorithms for the case of a dynamic dictionary, where patterns can be added to or removed from $\D$. We show how to process updates in $\cOtilde(n^\alpha)$ time and answer queries $\Dec(i,j)$, $\R(i,j)$ and $\RD(i,j)$ in $\cOtilde(n^{1-\alpha}+\occ)$ time for any $0<\alpha<1$, matching---up to subpolynomial factors---our conditional lower bound.

\paragraph{\bf Our techniques and a roadmap.}
First, in Section~\ref{sec:existsreport}, we present straightforward solutions for queries $\Dec(i,j)$ and $\R(i,j)$.
In Section~\ref{sec:rd} we describe an involved solution for $\RD(i,j)$ queries, that heavily relies on the periodic structure of the input text and on tools that we borrow from computational geometry.
In Section~\ref{sec:count} we rely on locally consistent parsing and further computational geometry tools to obtain an efficient solution for $\C(i,j)$ queries.
In Section~\ref{sec:dynamic} we extend our solutions for the case of a dynamic dictionary and provide a matching conditional lower bound.
Finally, in Appendix~\ref{sec:CD} we consider yet another type of queries, $\CD(i,j)$, and develop an approximate solution for them.

\section{Preliminaries}

We begin with basic definitions and notation generally following~\cite{AlgorithmsOnStrings}.
Let $T=T[1]T[2]\cdots T[n]$ be a \textit{string} of length $|T|=n$ over a linearly sortable alphabet $\Sigma$. The elements of $\Sigma$ are called \textit{letters}.
By $\varepsilon$ we denote an {\em empty string}.
For two positions $i$ and $j$ on $T$, we denote by $T[i\dd j]=T[i]\cdots T[j]$ the \textit{fragment} (sometimes called substring) of $T$ that starts at position $i$ and ends at position $j$ (it equals $\varepsilon$ if $j<i$). It is called \textit{proper} if $i>1$ or $j<n$. A fragment of $T$ is represented in $\cO(1)$ space by specifying the indices $i$ and $j$.
A {\em prefix} of $T$ is a fragment that starts at position $1$ ($T[1\dd j]$, notation: $T^{(j)}$) 
and a {\em suffix} is a fragment that ends at position $n$ ($T[i\dd n]$, notation: $T_{(i)}$).
We denote the {\em reverse string} of $T$ by $T^R$, i.e. $T^R=T[n]T[n-1]\cdots T[1]$.

Let $U$ be a string of length $m$ with $0<m\leq n$. 
We say that there exists an \textit{occurrence} of $U$ in $T$, or, more simply, that $U$ \textit{occurs in} $T$, when $U$ is a fragment of $T$.
We thus say that $U$ occurs at the \textit{starting position} $i$ in $T$ when $U=T[i \dd i + m - 1]$.

If a string $U$ is both a proper prefix and a proper suffix of a string $T$ of length~$n$, then $U$ is called a \emph{border} of $T$. A positive integer $p$ is called a \emph{period} of $T$ if $T[i] = T[i + p]$ for all $i = 1, \ldots, n - p$. A string $T$ has a period $p$ if and only if it has a border of length $n-p$. We refer to the smallest period as \emph{the period} of the string, and denote it as $\per(T)$, and, analogously, to the longest border as \emph{the border} of the string. A string is called \emph{periodic} if its period is no more than half of its length and \emph{aperiodic} otherwise.

The elements of the dictionary $\D$ are called \textit{patterns}. Henceforth we assume that $\varepsilon \not\in \D$, i.e.~the length of each $P \in \D$ is at least $1$. If $\varepsilon$ was in $\D$, we could trivially treat it individually. We further assume that each pattern of $\D$ is given by the starting and ending positions of its occurrence in $T$. Thus, the size of the dictionary $d=|\D|$ refers to the number of strings in $\D$ and not their total length.

The \textit{suffix tree} $\stt$ of a non-empty string $T$ of length $n$ is a compact trie representing all suffixes of $T$. The \textit{branching} nodes of the trie as well as the \textit{terminal} nodes, that correspond to suffixes of $T$, become {\em explicit} nodes of the suffix tree, while the other nodes are {\em implicit}.
Each edge of the suffix tree can be viewed as an upward maximal path of implicit nodes starting with an explicit node. Moreover, each node belongs to a unique path of that kind. Thus, each node of the trie can be represented in the suffix tree by the edge it belongs to and an index within the corresponding path.
We let $\mathcal{L}(v)$ denote the \textit{path-label} of a node $v$, i.e., the concatenation of the edge labels along the path from the root to $v$. We say that $v$ is  path-labelled  $\mathcal{L}(v)$. Additionally, $\delta(v)= |\mathcal{L}(v)|$ is used to denote the \textit{string-depth} of node~$v$. A terminal node $v$ such that $\mathcal{L}(v) = T_{(i)}$ for some $1 \leq i \leq n$ is also labelled with index~$i$. Each fragment of $T$ is uniquely represented by either an explicit or an implicit node of $\stt$, called its \emph{locus}.
Once $\stt$ is constructed, it can be traversed in a depth-first manner to compute the string-depth $\delta(v)$ for each explicit node~$v$.
The suffix tree of a string of length~$n$, over an integer ordered alphabet, can be computed in time and space $\cO(n)$~\cite{Farach-Colton:2000:SST:355541.355547}. In the case of integer alphabets, in order to access the child of an explicit node by the first letter of its edge label in $\cO(1)$ time, perfect hashing~\cite{DBLP:journals/jacm/FredmanKS84} can be used.
Throughout the paper, when referring to the suffix tree $\stt$ of $T$, we mean the suffix tree of $T\$$, where $\$ \not\in \Sigma$ is a sentinel letter that is lexicographically smaller than all the letters in $\Sigma$. This ensures that all terminal nodes are leaves.

We say that a tree is a \textit{weighted tree} if it is a rooted tree with an integer weight on each node $v$, denoted by $\omega(v)$, such that the weight of the root is zero and $\omega(u) < \omega(v)$ if $u$ is the parent of $v$. We say that a node $v$ is a \textit{weighted ancestor at depth $\ell$} of a node $u$ if $v$ is the highest ancestor of $u$ with weight of at least $\ell$.
After $\cO(n)$-time preprocessing, weighted ancestor queries for nodes of a weighted tree $\mathcal{T}$ of size $n$ can be answered in $\cO(\log \log n)$ time per query~\cite{DBLP:journals/talg/AmirLLS07}.
If $\omega$ has a property that the difference of weights of a child and its parent is always equal to 1, then the queries can be answered in $\cO(1)$ time after $\cO(n)$-time preprocessing~\cite{DBLP:journals/tcs/BenderF04}; in this special case the values $\omega$ are called \emph{levels} and the queries are called \emph{level ancestor queries}.
The suffix tree $\stt$ is a weighted tree with $\omega=\delta$.
Hence, the locus of a fragment $T[i \dd j]$ in $\stt$ is the weighted ancestor of the terminal node with path-label $T_{(i)}$ at string-depth $j-i$.

\section{$\Dec(i,j)$ and $\R(i,j)$ queries}\label{sec:existsreport}

We first present a convenient modification to the suffix tree with respect to a dictionary $\D$; see Fig.~\ref{fig:dmod-example}.

\begin{definition}
A \emph{$\D$-modified suffix tree} of a string $T$ is a tree with terminal nodes corresponding to non-empty suffixes of $T\$$ and branching nodes corresponding to $\{\varepsilon\} \cup \D$. A node corresponding to string $U$ is an ancestor of a node corresponding to string $V$ if and only if $U$ is a prefix of $V$. Each node stores its level as well as its string-depth (i.e., the length of its corresponding string).
\end{definition}
\begin{figure}[h!]
  \centering
  \begin{tikzpicture}[scale=0.75, transform shape]

    \tikzstyle{v}=[draw, circle, minimum width=0.75pt, inner sep=0.75pt, fill=black!80!white];
    \tikzstyle{marked}=[draw, shape=circle, minimum width=1pt, inner sep=2pt, red];
    \tikzstyle{label}=[midway,inner sep=1pt,above, sloped, fill=none];
    \tikzstyle{label2}=[midway,inner sep=1pt];
 

    \node[v] (r0) at (5, 4) {};

    \node[v] (a) at (0, 3) {};
    \node[v] (b) at (4.8, 3) {};
    \node[v] (c) at (6.1, 3) {}; \node[below of=c, node distance=8pt] {14};
    \node[v] (d) at (8, 3) {}; \node[below of=d, node distance=8pt] {2};

    \node[v] (aa) at ($(a.center)+(-2, -1)$) {};
    \node[v] (ab) at ($(a.center)+(0, -1)$) {};
    \node[v] (ac) at ($(a.center)+(1, -1)$) {}; \node[below of=ac, node distance=8pt] {13};
    \node[v] (ad) at ($(a.center)+(3, -1)$) {}; \node[below of=ad, node distance=8pt] {1};

    \node[v] (aaa) at ($(aa.center)+(-1, -1)$) {};
    \node[v] (aab) at ($(aa.center)+(0.1, -1)$) {};
    \node[v] (aac) at ($(aa.center)+(0.7, -1)$) {}; \node[below of=aac, node distance=8pt] {12};

    \node[marked] (aaaa) at ($(aaa.center)+(-0.2, -0.5)$) {};
    \node[v] (aaaa2) at ($(aaaa.center)+(-0.2, -1.7)$) {}; \node[below of=aaaa2, node distance=8pt] {3};
    \node[v] (aaab) at ($(aaa.center)+(0.2, -2.2)$) {}; \node[below of=aaab, node distance=8pt] {4};

    \node[v] (aaba) at ($(aab.center)+(-0.4, -2.2)$) {}; \node[below of=aaba, node distance=8pt] {5};
    \node[v] (aabb) at ($(aab.center)+(0.4, -2.2)$) {}; 
    \node[below of=aabb, node distance=8pt] {8};

    \node[v] (aba) at ($(ab.center)+(-0.75, -1.5)$) {}; \node[below of=aba, node distance=8pt] {6};
    \node[marked] (abb) at ($(ab.center)+(0.75, -0.75)$) {};
    \node[v] (abb2) at ($(abb.center)+(0, -0.75)$) {}; \node[below of=abb2, node distance=8pt] {9};

    \node[v] (ba) at ($(b.center)+(-0.75, -1)$) {};
    \node[v] (bb) at ($(b.center)+(0.75, -1)$) {}; \node[below of=bb, node distance=8pt] {10};

    \node[v] (baab) at ($(ba.center)+(-0.5, -1)$) {}; \node[below of=baab, node distance=8pt] {7};
    \node[v] (baac) at ($(ba.center)+(0.5, -1)$) {}; \node[below of=baac, node distance=8pt] {11};

    \draw (r0) -- (a) node[label] {$a$};
    \draw (r0) -- (b) node[label2, left, inner sep=3pt] {$b$};
    \draw (r0) -- (c) node[label, inner sep=0pt] {$c\$$};
    \draw (r0) -- (d) node[label] {$daaaabaabbaac\$$};

    \draw (a) -- (aa) node[label2, left, inner sep=4pt] {$a$};
    \draw (a) -- (ab) node[label2, left] {$b$};
    \draw (a) -- (ac) node[label, inner sep=0pt] {$c\$$};
    \draw (a) -- (ad) node[label] {$daaaabaabbaac\$$};

    \draw (aa) -- (aaa) node[label2, left] {$a$};
    \draw (aa) -- (aab) node[label2, right] {$b$};
    \draw (aa) -- (aac) node[label] {$c\$$};

    \draw (aaa) -- (aaaa) node[label] {$a$};
    \draw (aaaa) -- (aaaa2) node[label] {$baabbaac\$$};
    \draw (aaa) -- (aaab) node[label] {$baabbaac\$$};

    \draw (aab) -- (aaba) node[label] {$aabbaac\$$};
    \draw (aab) -- (aabb) node[label] {$baac\$$};

    \draw (ab) -- (aba) node[label] {$aabbaac\$$};
    \draw (ab) -- (abb) node[label] {$ba$};
    \draw (abb) -- (abb2) node[label] {$ac\$$};

    \draw (b) -- (ba) node[label] {$aa$};
    \draw (b) -- (bb) node[label] {$baac\$$};

    \draw (ba) -- (baab) node[label] {$bbaac\$$};
    \draw (ba) -- (baac) node[label, inner sep=0pt] {$c\$$};

    \node[marked] at (r0) {};
    \node[marked] at (c) {};
    \node[marked] at (aa) {};

\end{tikzpicture}
  \begin{tikzpicture}[scale=0.75, transform shape]

  \tikzstyle{v}=[draw, circle, minimum width=0.75pt, inner sep=0.75pt, fill=black!80!white];
  \tikzstyle{marked}=[draw, shape=circle, minimum width=1pt, inner sep=2pt, red];
  \tikzstyle{label}=[midway,inner sep=1pt,above, sloped, fill=none];
  \tikzstyle{label2}=[midway,inner sep=1pt];

  \node[marked] (r) at (5, 8) {};
  \node[left of=r, node distance=8pt] {$\varepsilon$};

  \foreach \d/\label in {0.3/6, 4/13, 5/1, 6/7, 7/11, 8/10, 10/2} {
    \node (r\label) at ($(r.center)+(\d, -3)$) {\label};
    \draw (r) -- (r\label);
  }

  \node[marked] (abba) at ($(r.center)+(1, -1)$) {};
  \draw (r) -- (abba);
  \node[left of=abba, node distance=13pt] {$abba$};
  \node (abba2) at ($(abba.center)+(0.5, -2)$) {9};
  \draw (abba) -- (abba2);

  \node[marked] (c) at ($(r.center)+(9, -3)$) {};
  \draw (r) -- (c);
  \node [below of=c, node distance=10pt] {14};
  \node[right of=c, node distance=8pt] {$c$};

  \node[marked] (aa) at ($(r.center)+(-2, -1)$) {};
  \node[left of=aa, node distance=10pt] {$aa$};

  \foreach \d/\label in {-1.5/4,-0.5/5, 0.5/8, 1.5/12} {
    \node (aa\label) at ($(aa.center)+(\d, -2)$) {\label};
    \draw (aa) -- (aa\label);
  }

  \node[marked] (aaaa) at ($(aa.center)+(-2, -1)$) {};
  \node[left of=aaaa, node distance=15pt] {$aaaa$};
  \node (aaaa2) at ($(aaaa.center)+(-1, -1)$) {3};

  \draw (r) -- (aa);
  \draw (aa) -- (aaaa);
  \draw (aaaa) -- (aaaa2);

\end{tikzpicture}
  \caption{Example of a $\D$-modified suffix tree
  for dictionary $\D=\{aa,aaaa,abba,c\}$ and text $T=adaaaabaabbaac$
  from Example~\ref{ex:1}.
  Top: the suffix tree of $T$ with the nodes corresponding to elements of $\D$ annotated in red;
  bottom: the $\D$-modified suffix tree of $T$.}
  \label{fig:dmod-example}
\end{figure}

\begin{lemma}\label{lem:dmodst}
A $\D$-modified suffix tree of $T$ has size $\cO(n+d)$ and can be constructed in $\cO(n+d)$ time.
\end{lemma}
\begin{proof}
The $\D$-modified suffix tree is obtained from the suffix tree $\stt$ in two steps.

In the first step, we mark all nodes of $\stt$ with path-label equal to a pattern $P \in \D$: if any of them are implicit, we first make them explicit; see Fig.~\ref{Fig:suffixTree}(a).
We can find the loci of the patterns in $\stt$ in $\cO(n+d)$ time by answering the weighted ancestor queries as a batch~\cite{DBLP:journals/corr/abs-1107-2422}, employing a data structure for a special case of Union-Find~\cite{DBLP:journals/jcss/GabowT85}. 
(If many implicit nodes along an edge are to become explicit, we can avoid the local sorting based on depth if we sort globally in time $\cO(n+d)$ using bucket sort and then add the new explicit nodes in decreasing order with respect to depth.)

In the second step, we recursively contract any edge $(u,v)$, where $u$ is the $\text{parent}$ of $v$ if:
\begin{enumerate}[itemsep=0ex, parsep=1pt, topsep=1pt]
\item both $u$ and $v$ are unmarked, or
\item $u$ is marked and $v$ is an unmarked internal node.
\end{enumerate}
The resulting tree is the $\D$-modified suffix tree and has $\cO(n)$ terminal nodes and $\cO(d)$ internal nodes; see Fig.~\ref{Fig:suffixTree}(b). 
\end{proof}

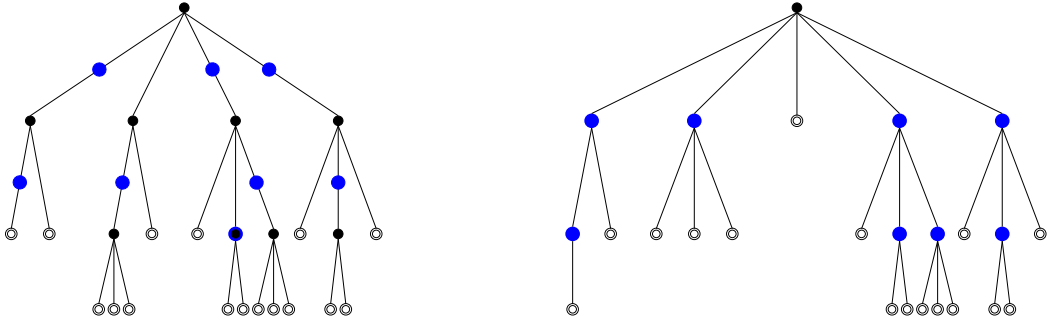
\begin{figure*}[htpb]
\captionsetup[subfigure]{font=scriptsize,labelfont=scriptsize}
\begin{subfigure}[t]{0.49\textwidth}
  \begin{center}
\resizebox{!}{0.50\totalheight}
{
\begin{tikzpicture}[-,>=stealth',level/.style={level 1/.style={sibling distance=2.7cm},
  level 2/.style={sibling distance=1.0cm},
  level 3/.style={sibling distance=0.4cm,level distance = 2cm},level distance = 3cm}] 
\node [arn_r] (tree){}
  child{node [arn_r](L1-C1){}
  child{node [arn_t](L2-C1){} edge from parent [pattern] node[left]{}}  
  child{node [arn_t](L2-C5){} edge from parent node[right]{}}
  edge from parent [pattern] node[left]{}
}	
child{node [arn_r](L1-C2){}
  child{node [arn_r](L3-C6){} 
  child{node [arn_t]{} edge from parent node[left]{}
  }
  child{node [arn_t]{} edge from parent node[above]{}
  }
  child{node [arn_t]{} edge from parent node[right]{}
  }
  edge from parent [pattern] node[left]{}
  }
  child{node [arn_t]{} edge from parent node[right]{}
  }
  edge from parent node[left]{}
}
child{node [arn_r](L1-C3){}
child{node [arn_t]{} edge from parent node[left]{}}
child{node [arn_p](vnode){}
 child{node [arn_t]{} edge from parent node[left]{}}
child{node [arn_t]{} edge from parent node[above]{}}
edge from parent node[ above] {}}  
child{node [arn_r](L3-C2){}
child{node [arn_t]{} edge from parent node[above]{}}   
child{node [arn_t]{} edge from parent node[above]{}}  
child{node [arn_t]{} edge from parent node[right]{}}
edge from parent [pattern] node[right]{}}
edge from parent [pattern] node[right]{}}
child{node [arn_r](L1-C4){}
child{node [arn_t] {} edge from parent node[left]{}}
child{node [arn_r](L3-C4){}
child{node [arn_t]{} edge from parent node[right]{}}
child{node [arn_t]{} edge from parent node[right]{}}
edge from parent[pattern]  node[right]{}}
child{node [arn_t] {} edge from parent node[left]{}}
edge from parent [pattern] node[right]{}
};
\end{tikzpicture}
}\caption{Seven implicit nodes of $\mathcal{T}(T)$ are made explicit. Each marked node (big circles in blue) represents a pattern $P\in \D$.}
\end{center}
\end{subfigure}
~
\begin{subfigure}[t]{0.49\textwidth}
\begin{center}
\resizebox{!}{0.50\totalheight}
{
\begin{tikzpicture}[-,>=stealth',level/.style={level 1/.style={sibling distance=2.7cm},
  level 2/.style={sibling distance=1cm},
  level 3/.style={sibling distance=0.4cm, level distance = 2cm},level distance = 3cm}] 
\node [arn_r] (tree){}
  child{node [arn_p1](L1-C1){}
  child{node [arn_p1](L2-C1){} 
  child{node [arn_t](L3-C1){} edge from parent  node[left]{} }  
  edge from parent node[right]{}}
  child{node [arn_t](L2-C5){} edge from parent node[right]{}}
  edge from parent node[right]{}}	
child{node [arn_p1](L1-C2){}
  child{node [arn_t]{} edge from parent node[right]{}
  }
  child{node [arn_t]{} edge from parent node[right]{}
  }
  child{node [arn_t]{} edge from parent node[right]{}
  }
  edge from parent node[left]{}
  }
child{node [arn_t]{} edge from parent node[right]{}
  }
child{node [arn_p1](L1-C3){}
child{node [arn_t]{} edge from parent node[right]{}}
child{node [arn_p1](vnode){}
child{node [arn_t]{} edge from parent node[left]{}}
child{node [arn_t]{} edge from parent node[above]{}}
edge from parent node[ above] {}}  
child{node [arn_p1](vnode){}
child{node [arn_t]{} edge from parent node[left]{}}
child{node [arn_t]{} edge from parent node[left]{}}
child{node [arn_t]{} edge from parent node[above]{}}
edge from parent node[ above] {}}  
edge from parent node[right]{}}
child{node [arn_p1](L1-C4){}
child{node [arn_t] {} edge from parent node[left]{}}
child{node [arn_p1](L3-C4){}
child{node [arn_t]{} edge from parent node[right]{}}
child{node [arn_t]{} edge from parent node[right]{}}
edge from parent node[right]{}}
child{node [arn_t] {} edge from parent node[left]{}}
edge from parent node[right]{}
};
\end{tikzpicture}
}\caption{$\D$-modified suffix tree.}
\end{center}
\end{subfigure}
\caption{The two-step construction of the $\D$-modified suffix tree.}\label{Fig:suffixTree}
\end{figure*}

We state the following simple lemma.

\begin{lemma}\label{lem:fwd_search}
With the $\D$-modified suffix tree of $T$ at hand, given positions $a,j$ in $T$ with $a\leq j$, we can compute all $P \in \D$ that occur at position $a$ and are of length at most $j-a+1$ in time $\cO(1+\occ)$.
\end{lemma}
\begin{proof}
We start from the root of the $\D$-modified suffix tree and go down towards the terminal node with path-label $T_{(a)}$. We report all encountered nodes $v$ as long as $\delta(v) \leq j-a+1$ is satisfied. We stop when this inequality is not satisfied.
\end{proof}

The $\D$-modified suffix tree enables us to answer $\Dec(i,j)$ and $\R(i,j)$ queries. 

\begin{theorem}\label{thm:easy}~

\vspace*{-0.2cm}
\begin{enumerate}[label=(\alph*)]
\item\label{thm:exists}
$\Dec(i,j)$ queries can be answered in $\cO(1)$ time with a data structure of size $\cO(n)$ that can be constructed in $\cO(n+d)$ time.
\item\label{thm:report}
$\R(i,j)$ queries can be answered in $\cO(1+\occ)$ time with a data structure of size $\cO(n+d)$ that can be constructed in $\cO(n+d)$  time.
\end{enumerate}
\end{theorem}
\begin{proof}
\ref{thm:exists} Let us define an array $B[a]=\min \{ b : T[a\dd b] \in \D\}$. If there is no pattern from $\D$ starting
in $T$ at position $a$, then $B[a]=\infty$.
It can be readily verified that the answer to query $\Dec(i,j)$ is yes if and only if the minimum element in the subarray $B[i\dd j]$ is at most $j$.
Thus, in order to answer $\Dec(i,j)$ queries, it suffices to construct the array $B$ and a data structure that answers range minimum queries (\RMQ{}) on $B$.
Using the $\D$-modified suffix tree of $T$, whose construction time is the bottleneck, array $B$ can be populated in $\cO(n)$ time as follows.
For each terminal node with path-label $T_{(a)}$ and level greater than $1$, we set $B[a]$ to the string-depth of its ancestor at level $1$ using a level ancestor query.
If the terminal node is at level $1$, then $B[a]=\infty$.
A data structure answering range minimum queries in $\cO(1)$ time can be built in time $\cO(n)$~\cite{DBLP:journals/siamcomp/HarelT84,DBLP:journals/jal/BenderFPSS05}.

\ref{thm:report} We first identify all positions $a \in [i \dd j]$ that are starting positions of occurrences of some pattern $P\in \D$ in $T[i \dd j]$ using \RMQ{}s over array $B$, which has been defined in the proof of part \ref{thm:exists}, as follows.
The first \RMQ{}, is over the range $[i \dd j]$ and identifies a position $a$ (if any such position exists).
The range is then split into two parts, namely $[i, a-1]$ and $[a+1,j]$. We recursively, use \RMQ{}s to identify the remaining positions in each part. Once we have found all the positions where at least one pattern from $\D$ occurs, we report all the patterns occurring at each of these positions and being contained in $T[i \dd j]$.
The complexities follow from Lemmas~\ref{lem:dmodst} and~\ref{lem:fwd_search}.
\end{proof}


\section{$\RD(i,j)$ queries}\label{sec:rd}

Below, we present an algorithm that reports patterns from $\D$ occurring in $T[i \dd j]$, allowing for $\Oh(1)$ copies of each pattern on the output. We can then sort these patterns, remove duplicates, and report distinct ones using an additional global array of counters, one for each pattern.

Let us first partition $\D$ into $\D_0,\ldots ,\D_{\lfloor \log n \rfloor}$ such that $\D_k=\{P\in \D : \floor{\log |P|}=k\}$. We call $\D_k$ a \emph{$k$-dictionary}. We now show how to process a single $k$-dictionary $\D_k$; the query procedure may clearly assume $k \leq \log|T[i\dd j]|$.

We precompute an array $L_k[1 \dd n]$ such that $T[a \dd L_k[a]]$ is the longest pattern in $\D_k$ is a prefix of $T_{(a)}$. We can do this in $\cO(n)$ time by inspecting the parents of terminal nodes in the $\D_k$-modified suffix tree. Next, we assign to all the patterns of $\D_k$ equal to some $T[a \dd a+L_k[a]]$ integer identifiers $\id$ (or colors) in $[1 \dd n]$, and construct an array $I_k[a]=\id(P)$, where $P=T[a \dd a+L_k[a]]$. We then rely on the following theorem.

\begin{theorem}[Colored Range Reporting~\cite{DBLP:conf/soda/Muthukrishnan02}]\label{thm:muthu}
Given an array $A[1\dd N]$ of elements from $[1 \dd U]$, we can construct a data structure of size $\cO(N)$ in $\cO(N + U)$ time, so that upon query $[i\dd j]$ 
 all distinct elements in $A[i\dd j]$ can be reported in $\cO(1+\occ)$ time.
\end{theorem}

We first perform a colored range reporting query on the range $[i \dd j-2^{k+1}]$ of array $I_k$ and obtain a set of distinct patterns $\mathcal{C}_k$, employing Theorem~\ref{thm:muthu}. We observe the following.

\begin{observation}
Any pattern of a $k$-dictionary $\D_k$ occurring in $T$ at position $p\in [i\dd j-2^{k+1}]$ is a prefix of a pattern~$P\in \mathcal{C}_k$.
\end{observation}

Based on this observation, we will report the remaining patterns using the $\D_k$-modified suffix tree, following parent pointers and temporarily marking the loci of reported patterns to avoid double-reporting.
We thus now only have to compute the patterns from $\D_k$ that occur in $T[t\dd j]$, where $t=\max\{i,j-2^{k+1}+1\}$.

We further partition $\D_k$  for $k>1$ to a \emph{periodic $k$-dictionary} and \emph{aperiodic $k$-dictionary}: \[\D_k^{p}=\{P\in \D_k : \per(P)\leq 2^k/3\}\quad\text{and}\quad\D_k^{a}=\{P\in \D_k : \per(P) > 2^{k}/3\}.\] Note that we can partition $\D_k$ in $\cO(|\D_k|)$ time using the so-called \textsc{2-Period Queries} of~\cite{DBLP:conf/soda/KociumakaRRW15,DBLP:journals/siamcomp/BannaiIINTT17,tomeksthesis}. Such a query decides whether a given fragment of the text is periodic and, if so, it also returns its period. It can be answered in $\cO(1)$ time after an $\cO(n)$-time preprocessing of the text.

\subsection{Processing aperiodic $k$-dictionary} 

We make use of the following sparsity property. 

\begin{fact}[Sparsity of occurrences] \label{fact:sparse}
The occurrences of a pattern $P$ of an aperiodic $k$-dictionary $\D_k^a$ in $T$ start over $\frac16|P|$ positions apart.  
\end{fact} 
\begin{proof}
  If two occurrences of $P$ started $d \le \frac{2^k}{3}$ positions apart, then $d$ would be a period of $P$, contradicting $P\in \D_k^a$.
Then, since $2^k \leq |P| < 2^{k+1}$, we have that $2^k/3 \geq \frac16|P|$.
\end{proof}

\begin{lemma}
$\RD(t,j)$ queries for the aperiodic $k$-dictionary $\D^a_k$ and $j-t \leq 2^{k+1}$ can be answered in $\cO(1+\occ)$ time with a data structure of size $\cO(n+|\D^a_k|)$, that can be constructed in $\cO(n+|\D^a_k|)$ time.
\end{lemma}
\begin{proof}
Since the fragment $T[t \dd j]$ is of length at most $2^{k+1}$, it may only contain a constant number of occurrences of each pattern in $\D_k^a$ by~\cref{fact:sparse}. We can thus simply use a $\R(t,j)$ query for dictionary $\D_k^a$ and then remove duplicates. The complexities follow from~Theorem~\ref{thm:easy}\ref{thm:report}.
\end{proof}

\subsection{Processing periodic $k$-dictionary}

Our solution for periodic patterns relies on the well-studied theory of maximal repetitions (\emph{runs}) in strings. A run is a periodic fragment $R=T[a\dd b]$ which can be extended neither to the left nor to the right without increasing the period $p= \per(R)$, that is, $T[a-1]\neq T[a +p-1]$ and $T[b-p+1] \neq T[b+1]$ provided that the respective letters exist.
The number of runs in a string of length $n$ is $\cO(n)$ and  all the runs  can be computed in $\cO(n)$ time \cite{DBLP:conf/focs/KolpakovK99,DBLP:journals/siamcomp/BannaiIINTT17}.

\begin{observation}\label{obs:run}
Let $P$ be a periodic pattern. If $P$ occurs in $T[t\dd j]$, then $P$ is a fragment of a unique run $R$ such that $\per(R) = \per(P)$. We say that this run $R$ \emph{extends} $P$.
\end{observation} 

 Let $\Ru$ be the set of all runs in $T$. Following~\cite{tomeksthesis}, we construct for all $k\in [0\dd \floor{\log n}]$ the sets of runs $\Ru_k=\{R \in \Ru : \per(R) \leq \frac{2^k}{3}, |R| \geq 2^k \}$ in $\cO(n)$ time overall. Note that these sets are not disjoint; however, $|\Ru_k|=\cO(\frac{n}{2^k})$ (cf.~\cref{lem:runsk} below) and thus their total size is $\cO(n)$. If $U$ is a fragment of $T$, by $\Ru_k(U)\subseteq \Ru_k$ we denote the set of all runs $R \in \Ru_k$ such that $|R\cap U| \ge 2^k$, that is, runs whose overlap with the frgment $U$ is at least $2^k$.

 \begin{lemma}[{see \cite[Lemma 4.4.7]{tomeksthesis}}]\label{lem:runsk}
  $|\Ru_k(U)| =\cO \big( \frac1{2^k}|U|\big)$.
 \end{lemma}
 
\paragraph{\bf Strategy.} Given a fragment $U=T[t \dd j]$, we will first identify all runs $\Ru_k(U)$ of $\Ru_k$ that have a sufficient overlap with $U$. There is a constant number of them by~\cref{lem:runsk}. For an occurrence of a pattern $P \in \D_k^p$ in $U$, the unique run $R$ extending this occurrence of $P$ must be in $\Ru_k(U)$. We will preprocess the runs in order to be able to compute a unique (the leftmost) occurrence \emph{induced} by run $R$ for each such pattern $P$.
  
\begin{lemma}\label{lem:runsk2}
 Let $U$ be a fragment of $T$ of length at most $2^{k+1}$. Then $\Ru_k(U)$ can be retrieved in $\cO(1)$ time after an $\cO(n)$-time preprocessing.
\end{lemma}
\begin{proof}
\textsc{Periodic Extension Queries}~\cite[Section 5.1]{tomeksthesis}, given a fragment $V$ of the text $T$ as input, return the run $R$ extending $V$. They can be answered in $\cO(1)$ time after $\cO(n)$-time preprocessing.

Let us cover $U$ using $\Oh(\frac{1}{2^k}|U|)$ fragments of length $\frac{2^{k+1}}{3}$ with overlaps of at least $\frac{2^k}{3}$ and ask a \textsc{Periodic Extension Query} for each fragment $V$ in the cover. For each run $R \in \Ru_k(U)$ with sufficient overlap, $R\cap U$ must contain a fragment $V$ in the cover and its periodic extension must be $R$ since $|V|\ge 2\cdot \per(R)$.
\end{proof}

\paragraph{\bf Preprocessing.}
We construct an array $\ell_k[1\dd n]$ such that $T[i\dd \ell_k[i]]$ is the shortest pattern $P\in \D_k^p$
that occurs at position $i$. 
Note that $\ell_k[i]$ can be retrieved in $\Oh(1)$ time using a level ancestor query in the $\D_p^k$-modified suffix tree
(asking for a level-1 ancestor of the leaf corresponding to $T_{(i)}$, as in the proof of Theorem~\ref{thm:easy}\ref{thm:exists}).
We then preprocess the array $\ell_k$ for \RMQ queries. 

\paragraph{\bf Processing a run at query.} Let us begin with a consequence of the fact that the shortest period is primitive.

\begin{observation}\label{obs:firstp}
  If a pattern $P$ occurs in a text $Q$ and satisfies $|P|\ge \per(Q)$, then $P$ has exactly one occurrence in the first $\per(Q)$ positions of $Q$.
\end{observation}

We use \RMQ{}s repeatedly, as in the proof of~Theorem~\ref{thm:easy}\ref{thm:report}, for the subarray of $\ell_k$ corresponding to the first $\per(R)$ positions of $R \cap U$.
This way, due to~\cref{obs:firstp}, we compute exactly the positions where a pattern $P\in \D_k^p$ has its leftmost occurrence in $R\cap U$.
The number of positions identified for a single run $R\in \Ru_k(U)$ is therefore upper bounded by the number of distinct patterns
occurring within $R\cap U$.
We then report all distinct patterns occurring within $R\cap U$ by processing each such starting position using~\cref{lem:fwd_search}.
There is no double-reporting while processing a single run, by~\cref{obs:firstp} and hence the time required to process this run is $\cO(1+\occ)$ --- $\occ$ here refers to the number of distinct patterns from $\D_k^p$ occurring within $U$.
Since $|\Ru_k(U)|=\Oh(1)$, we report each pattern a constant number of times and the overall time required is $\cO(1+\occ)$.

\subsection{Reducing the space}
The space occupied by our data structure can be reduced to $\cO(n+d)$.
We store the $\D$-modified suffix tree and mark all nodes from each $\D_k^i$, for $k\in [0\dd \floor{\log n}])$ and $i \in \{a,p\}$,
with a different color. For each dictionary $\D'=\D_k^i$, we further store, using $\cO(|\D'|)$ space, the $\D'$-modified suffix tree without its unmarked leaves.

We will show below that the only additional operation we now need to support is determining the parent of a given unmarked leaf in the original $\D'$-modified suffix tree (before the leaves were chopped).
This can be done using the nearest colored ancestor data structure of~\cite{DBLP:journals/tcs/GawrychowskiLMW18} over the $\D$-modified suffix tree. For a tree of size $N$, it achieves $\cO(\log\log N)$ time per query after $\cO(N)$-time preprocessing. We can, however, exploit the fact that we only have $\textsf{colors} = \cO(\log n)$ colors to obtain constant-time queries within the same construction time.

In~\cite{DBLP:journals/tcs/GawrychowskiLMW18} it is shown that, in order to answer nearest colored ancestor queries in a tree with $N$ nodes, it is enough to store some arrays of total size $\cO(N)$ and predecessor data structures for $\cO(\textsf{colors})$ subsets of $[1 \dd 2N]$ whose total size is $\cO(N)$. The time needed to compute the sets for the predecessor data structures and the arrays is $\cO(N)$.
The time complexity of the query is proportional to the time required for answering a constant number of predecessor queries over the aforementioned sets.
We implement a predecessor data structure for a set $S \subseteq [1 \dd 2N]$ using $\cO(N)$ bits of space as follows. We store a bitmap that has the $i$th bit set if and only if $i \in S$ and augment it with a data structure that answers \textsf{rank} and \textsf{select} queries in $\cO(1)$ time and requires $o(N)$ additional bits of space~\cite{DBLP:conf/focs/Jacobson89,rsthesis}. Such a component can be constructed in $\cO(N/ \log N)$ time~\cite{DBLP:conf/soda/BabenkoGKS15,DBLP:journals/tcs/MunroNV16}. Note that $\textsf{pred}_S(i)=\textsf{select}(\textsf{rank}(i))$.
We thus use $\cO((n+d) \log n)$ bits, i.e., $\Oh(n+d)$ machine words, in total for the part of the data structure responsible for reporting occurrences starting at given positions.

It was shown in~\cite{DBLP:journals/siamcomp/FischerH11} that we can implement an $\cO(1)$-query-time \RMQ data structure for an array of size $N$ using $\Oh(N)$ bits. This data structure only returns the index of the minimum value in the given range.

For colored range reporting, the main component of the data structure underlying~\cref{thm:muthu} from~\cite{DBLP:conf/soda/Muthukrishnan02} is an \RMQ{} data structure over array $J[i]=\max \{j : j<i, \: A[i]=A[j]\}$.
We build an $\cO(N)$-bits \RMQ{} data structure over $J$.
The query procedure however, needs access to $A$, i.e.~the colors. 
We can retrieve the value of the $i$-th element in our array of colors using our representation of the $\D'$-modified suffix tree, since the its color corresponds to the respective leaf in the $D'$-modified suffix tree or, if it is unmarked, to that of its parent.

Then, filtering $\Oh(\occ)$ starting positions in the periodic case, is based on \RMQ{} queries over multiple arrays of total length $\Oh(n\log n)$.
To construct them, we build the $\D'$-modified suffix trees one by one and build the relevant \RMQ{} data structures that require $\Oh(n\log n)$ bits in total before chopping the unmarked leaves. 
The actual value at the indices returned by \RMQ{} queries can, as above, be determined using our representation of the $\D'$-modified suffix tree.

We arrive at the main result of this section.

\begin{theorem}
$\RD(i,j)$ queries can be answered in $\cO(\log n + \occ)$ time with a data structure of size $\cO(n+d)$ that can be constructed in $\cO(n \log n+d)$ time.
\end{theorem}

\section{$\C(i,j)$ queries}\label{sec:count}

We first solve an auxiliary problem and show how it can be employed to give an unsatisfactory solution to $\C(i,j)$. We then refine our approach using recompression and obtain the following.

\begin{theorem}
$\C(i,j)$ queries can be answered in $\cO(\log^2 n/\log\log n)$ time with a data structure of size $\cO(n+d \log n)$ that can be constructed in $\cO(n\log n / \log \log n +d\log^{3/2} n)$ time.
\end{theorem}

\subsection{An auxiliary problem}\label{sec:auxprob}

By inter-position $i+1/2$ we refer to a location between positions $i$ and $i+1$ in $T$. We also refer to inter-positions $1/2$ and $n+1/2$.
We consider the following auxiliary problem, in which we are given a set of inter-positions  (\emph{breakpoints}) $B$ of $P$ and upon query we are to compute all fragments of $T[i \dd j]$ that align a specific inter-position (\emph{anchor}) $\beta$ of the text with some inter-position in $B$.

\defproblem{\textsc{Breakpoints-Anchor IPM}}{%
  A length-$n$ text $T$, its length-$m$ substring $P$, and a set $B$ of inter-positions (breakpoints) of $P$.
}{%
  $\C_{\beta}(i,j)$: the number of fragments $T[r \dd r+m-1]$ of $T[i \dd j]$ 
  that match $P$ such that $\beta-r+1 \in B$ ($\beta$ is an anchor).
}

In the 2D orthogonal range counting problem, one is to preprocess an $n \times n$ grid with $\cO(n)$ marked points so that upon query $[x_1,y_1]\times [x_2,y_2]$, the number of points in this rectangle can be computed efficiently. In the (dual) 2D range stabbing counting problem, one is to preprocess the grid with $\cO(n)$ rectangles so that upon query $(x,y)$ the number of (stabbed) rectangles that contain $(x,y)$ can be retrieved efficiently.
The counting version of range stabbing queries in 2D reduces to two-sided range counting queries in 2D as follows (cf.~\cite{DBLP:journals/siamcomp/Patrascu11}).
For each rectangle $[x_1,y_1]\times [x_2,y_2]$ in grid $G$, we add points $(x_1,y_1)$ and $(x_2+1,y_2+1)$ with weight $1$ and points $(x_1,y_2+1)$ and $(x_2,y_1+1)$ with weight $-1$ in a grid $G'$.
Then the number of rectangles stabbed by point $(a,b)$ in $G$ is equal to the sum of weights of points in $(-\infty,a] \times (-\infty,b]$ in $G'$.
We will use the following result in our solution to \textsc{Breakpoints-Anchor IDM} (Lemma~\ref{lem:aux_break1}).

\begin{theorem}[\cite{DBLP:journals/tcs/MunroNV16}]\label{thm:range_count}
Range counting queries for $n$ points in 2D (rank space) can be answered in time $\cO(\log n / \log \log n)$ with a data structure of size $\cO(n)$ that can be constructed in time $\cO(n \sqrt{\log n})$.
\end{theorem}

\paragraph{\bf Data structure.}
Let $W_1=\{P[\ceil{b} \dd m]: b \in B\}$ and consider the set $W_2$ obtained by adding $U\$$ and $U\#$ for each element $U$ of $W_1$ to an initially empty set, where $\$$ is a letter smaller (resp.~$\#$ is larger) than all the letters in $\Sigma$.
Let $W$ be the compact trie for the set of strings $W_2$.
For each internal node $v$ of $W$ that does not have an outgoing edge with label $\$$, we add such a (leftmost) edge with a leaf attached to its endpoint.
$W$ can be constructed in $\cO(|B|)$ time after an $\cO(n)$-time preprocessing of $T$, allowing for constant-time longest common prefix queries; cf.~\cite{AlgorithmsOnStrings}.
We also build the $W_1$-modified suffix tree of $T$ and preprocess it for weighted ancestor queries. We keep two-sided pointers between nodes of $W$ and of the $W_1$-modified suffix tree of $T$ that have the same path-label.
Similarly, let $W^R$ be the compact trie for set $Z_2$ consisting of elements $U\$$ and $U\#$ for each $U \in Z_1=\{(P[1 \dd \floor{b}])^R: b \in B\}$. We preprocess $W^R$ analogously.
Each of the tries has at most $k=\cO(|B|)$ leaves.

Let us now consider a 2D grid of size $k \times k$, whose $x$-coordinates (resp.~$y$-coordinates) correspond to the leaves of $W$ (resp.~$W^R$).
For each $b \in B$ we do the following.
Let $x_1$ and $x_2$ be the leaves with path-label $P[\ceil{b} \dd m]\$$ and $P[\ceil{b} \dd m]\#$ in $W$, respectively. Similarly, let $y_1$ and $y_2$ be the leaves with path-label $(P[1 \dd \floor{b}])^R\$$ and $(P[1 \dd \floor{b}])^R\#$ in $W^R$, respectively.
We add the rectangle $R_b=[x_1,y_1] \times [x_2,y_2]$ in the grid. An illustration is provided in Fig.~\ref{fig:rectangles}.
We then preprocess the grid for the counting version of 2D range stabbing queries, employing Theorem~\ref{thm:range_count}.

\begin{figure}[t!]
  \centering
  \scalebox{0.95}{\begin{tikzpicture}[scale=0.75, transform shape]

    \tikzstyle{v}=[draw, circle, minimum width=1pt, inner sep=1pt, fill=black!80!white];
    \tikzstyle{label}=[midway,inner sep=1pt,above,sloped];
    
    \begin{scope}[yshift=14.5cm]
    \node at (6, 4.5) {$P=abaabb$};

    \foreach \i in {0,...,15} {
        \node[v] (v\i) at (\i,0) {};
    }
    \node[v] (i1) at (4, 3) {};
    \node[v] (i2) at (11.5, 3) {};
    \node[v] (i3) at (2.5, 2) {};
    \node[v] (i4) at (6.5, 2) {};
    \node[v] (i5) at (10.5, 1.5) {};
    \node[v] (i6) at (12.5, 1.5) {};
    \node[v] (i7) at (5.5, 1) {};
    \node[v] (i8) at (7.5, 1) {};
    \node[v] (r0) at (8, 4) {};

    \draw (r0) -- (i1) node[label] {$a$};
    \draw (r0) -- (i2) node[label] {$b$};
    \draw (i1) -- (i3) node[label] {$abb$};
    \draw (i1) -- (i4) node[label] {$b$};
    \draw (i4) -- (i7) node[label] {$aabb$};
    \draw (i4) -- (i8) node[label] {$b$};
    \draw (i2) -- (i5) node[label] {$aabb$};
    \draw (i2) -- (i6) node[label] {$b$};

    \draw (i3) -- (v2) node[label] {$\$$};
    \draw (i3) -- (v3) node[label] {$\#$};
    \draw (i7) -- (v5) node[label] {$\$$};
    \draw (i7) -- (v6) node[label] {$\#$};
    \draw (i8) -- (v7) node[label] {$\$$};
    \draw (i8) -- (v8) node[label] {$\#$};
    \draw (i5) -- (v10) node[label] {$\$$};
    \draw (i5) -- (v11) node[label] {$\#$};
    \draw (i6) -- (v12) node[label] {$\$$};
    \draw (i6) -- (v13) node[label] {$\#$};

    \draw (r0) to [out=180,in=90] node[sloped,midway,above] {$\$$} (v0) ;
    \draw (r0) to [out=0,in=90] node[sloped,midway,above] {$\#$} (v15) ;

    \draw (i1) to [out=180,in=90] node[sloped,midway,above] {$\$$} (v1) ;

    \draw (i2) to [out=180,in=90] node[sloped,midway,above] {$\$$} (v9) ;
    \draw (i2) to [out=0,in=90] node[sloped,midway,above] {$\#$} (v14) ;

    \draw (i4) to [out=180,in=90] node[sloped,midway,above] {$\$$} (v4) ;
    \end{scope}
    
    \begin{scope}[xshift=-0.5cm, yshift=0cm, rotate=90, transform shape]
    \node at (6, 4.5) {$P^{R}=bbaaba$};
    \foreach \i in {0,...,14} {
        \node[v] (u\i) at (\i,0) {};
    }
    
    \node[v] (j1) at (3.5, 2.5) {};
    \node[v] (j2) at (11, 3) {};
    \node[v] (j3) at (2.5, 1) {};
    \node[v] (j4) at (4.5, 1) {};
    \node[v] (j5) at (9.5, 2) {};
    \node[v] (j6) at (12.5, 1.5) {};
    \node[v] (j7) at (9.5, 1) {};

    \node[v] (r1) at (6, 4) {};

    \draw (r1) -- (j1) node[label,rotate=180] {$a$};
    \draw (r1) -- (j2) node[label] {$b$};
    \draw (j1) -- (j3) node[label,rotate=180] {$aba$};
    \draw (j1) -- (j4) node[label] {$ba$};
    \draw (j2) -- (j5) node[label,rotate=180] {$a$};
    \draw (j2) -- (j6) node[label] {$baaba$};
    \draw (j5) -- (j7) node[label,rotate=180] {$aba$};

    \draw (j3) -- (u2) node[label] {$\$$};
    \draw (j3) -- (u3) node[label] {$\#$};
    \draw (j4) -- (u4) node[label] {$\$$};
    \draw (j4) -- (u5) node[label] {$\#$};
    \draw (j7) -- (u9) node[label] {$\$$};
    \draw (j7) -- (u10) node[label] {$\#$};
    \draw (j6) -- (u12) node[label] {$\$$};
    \draw (j6) -- (u13) node[label] {$\#$};

    \draw (r1) to [out=180,in=90] node[sloped,midway,above,rotate=180] {$\$$} (u0) ;
    \draw (r1) to [out=0,in=90] node[sloped,midway,above] {$\#$} (u14) ;

    \draw (j1) to [out=180,in=90] node[sloped,midway,above,rotate=180] {$\$$} (u1) ;
    \draw (j1) to [out=0,in=90] node[sloped,midway,above] {$\#$} (u6) ;
    \draw (j5) to [out=180,in=90] node[sloped,midway,above,rotate=180] {$\$$} (u8) ;
    \draw (j5) to [out=0,in=90] node[sloped,midway,above] {$\#$} (u11) ;
    \draw (j2) to [out=180,in=90] node[sloped,midway,above,rotate=180] {$\$$} (u7) ;

    \end{scope}
  
    \begin{scope}
    \draw [thick] (v5 |- u0) node[below,xshift=0.5cm]  {$\frac{1}{2}$} rectangle (v6 |- u14); 
    
    \draw[thick] (v10 |- u1) node[below,xshift=0.5cm]  {$1\frac{1}{2}$} rectangle (v11 |- u6);  
    \draw[thick] (v2 |- u8) node[below,xshift=0.5cm]  {$2\frac{1}{2}$} rectangle (v3 |- u11);  
    \draw[thick] (v7 |- u4) node[below,xshift=0.5cm]  {$3\frac{1}{2}$} rectangle (v8 |- u5);  
    \draw[thick] (v12 |- u2) node[below,xshift=0.5cm]  {$4\frac{1}{2}$} rectangle (v13 |- u3);  
    \draw[thick] (v9 |- u9) rectangle (v14 |- u10) node[right,yshift=-0.5cm]  {$5\frac{1}{2}$};  
    \draw[thick] (v0 |- u12) rectangle (v15 |- u13) node[right,yshift=-0.5cm]  {$6\frac{1}{2}$};  
    \end{scope}

  \end{tikzpicture}}
  \caption{
    Example of the construction of rectangles in the proof of \cref{lem:aux_break1} 
    for $P=\texttt{abaabb}$ and breakpoints $i+1/2$ for $i=0,1,2,3,4,5,6$.
    Each rectangle is annotated with its breakpoint.
  }
  \label{fig:rectangles}
\end{figure}
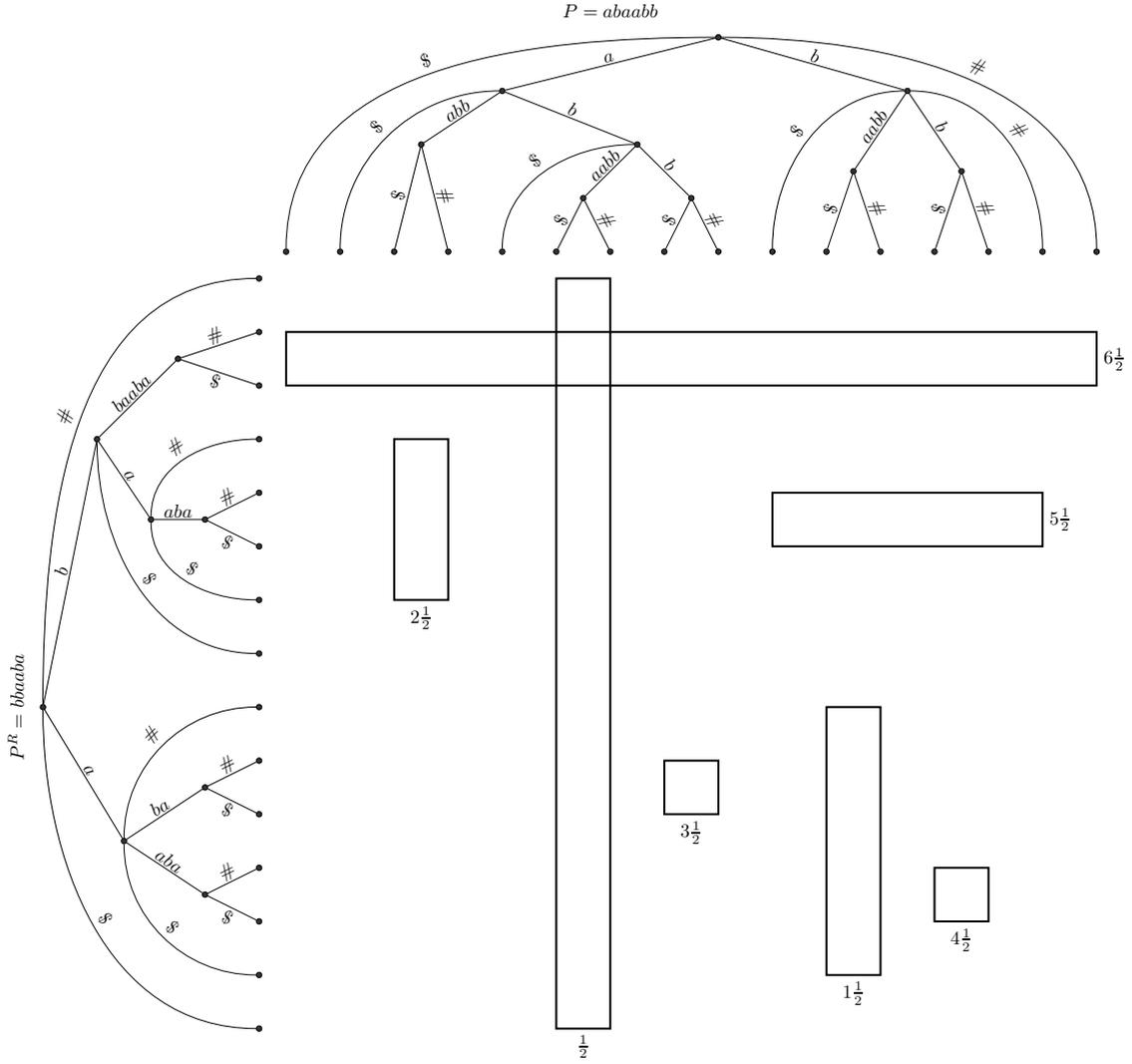

\paragraph{\bf Query.}
Let the longest prefix of $T[\ceil{\beta} \dd j]$ that is a prefix of an element of $W_1$ be $U$ and its locus in $W$ be $u$. This can be computed in $\cO(\log\log n)$ time using a weighted ancestor query in the $W_1$-modified suffix tree of $T$ and following the pointer to $W$.
If $u$ is an explicit node, we follow the edge with label $\$$, while if it is implicit along edge $(p,q)$, we follow the edge with label $\$$ from $p$. In either case, we reach a leaf $u'$. We do the symmetric procedure with $(T[i \dd \floor{\beta}])^R$ in $W^R$ and obtain a leaf $v'$.

\begin{observation}
The number of fragments $T[r \dd t]=P$ with $r,t \in [i \dd j]$ and $\beta-r+1 \in B$ is equal to the number of rectangles stabbed by the point  of the grid defined by $u'$ and $v'$.
\end{observation}

The observation holds because this point is inside rectangle $R_b$ for $b \in B$ if and only if $P[\ceil{b} \dd m]$ is a prefix of $T[\ceil{\beta} \dd j]$ and $P[1 \dd \floor{b}]$ is a suffix of $T[i \dd \floor{\beta}]$. This concludes the proof of the following result.

\begin{lemma}\label{lem:aux_break1}
  \textsc{Breakpoints-Anchor IPM} queries can be answered in time $\cO(\log n /\log\log n)$ with a data structure of size $\cO(n+|B|)$ that can be constructed in time $\cO(n+|B|\sqrt{\log |B|})$.
  \end{lemma}

For the analogously defined problem \textsc{Breakpoints-Anchor IDM}, we obtain the following lemma by building trie $W$ for the union of the sets $W_2$ defined in the above proof for each pattern (similarly for $W^R$) and adding all rectangles to a single grid.

\begin{lemma}\label{lem:aux_break2}
\textsc{Breakpoints-Anchor IDM} queries can be answered in time $\cO(\log n /\log\log n)$ with a data structure of size $\cO(n+ \sum_{P \in \D} |B_P|)$. The data structure can be constructed in time $\cO(n+ \sqrt{\log n} \sum_{P \in \D} |B_P|)$.
\end{lemma}

\paragraph{\bf A warm-up solution for $\C(i,j)$.} \cref{lem:aux_break2} can be applied naively to answer $\C(i,j)$ queries as follows.
Let us set $B_P=\{p+1/2 : p \in [1\dd |P|-1] \}$ for each pattern $P \in \D$ and construct the data structure of \cref{lem:aux_break2}.
We build a balanced binary tree $\BTr$ on top of the text and for each node $v$ in $\BTr$ define $\val(v)$ to be the fragment consisting of the characters corresponding to the leaves in the subtree of $v$.
Note that if $v$ is a leaf, then $|\val(v)|=1$; otherwise, $\val(v)=\val(u_\ell)\val(u_r)$, where $u_\ell$ and $u_r$ are the children of $v$.
For each node $v$ in $\BTr$, we precompute and store the count for $\val(v)$.
If $v$ is a leaf, this count can be determined easily. Otherwise, each occurrence is contained in $\val(u_\ell)$, is contained in $\val(u_r)$,
or spans both $\val(u_\ell)$ and $\val(u_r)$. Hence, we sum the answers for the children $u_\ell$ and $u_r$ of $v$
and add the result of a \textsc{Breakpoints-Anchor IDM} query in $\val(v)$ with the anchor between $\val(u_\ell)$ and $\val(u_r)$.

To answer a query concerning $T[i\dd j]$, we recursively count the occurrences in the intersection of $\val(v)$ with $T[i\dd j]$, starting from the root $r$ of $\BTr$ as it satisfies $\val(r)=T[1\dd n]$.
If the intersection is empty, the result is 0, and if $\val(v)$ is contained in $T[i\dd j]$, we can use the precomputed count.
Otherwise, we recurse on the children $u_\ell$ and $u_r$ of $v$ and sum the resulting counts. It remains to add the number of occurrences
spanning across both $\val(u_\ell)$ and $\val(u_r)$. This value is non-zero only if $T[i\dd j]$ spans both these fragments,
and it can be determined from a \textsc{Breakpoints-Anchor IDM} query in the intersection of $\val(v)$ and $T[i\dd j]$
with the anchor between $\val(u_\ell)$ and $\val(u_r)$.

The query-time is $\cO(\log^2 n /\log\log n)$ since non-trivial recursive calls are made only for nodes on the paths from the root $r$ to the leaves representing $T[i]$ and $T[j]$. Nevertheless, the space required for this ``solution'' can be $\Omega(nd)$, which is unacceptable.
Below, we refine this technique using a locally consistent parsing; our goal is to decrease the size of each set $B_P$ from $\Theta(|P|)$ to $\cO(\log n)$.

\subsection{Recompression}
A \emph{run-length straight line program} (\emph{RSLP}) is a context-free grammar which generates exactly one string
and contains two kinds of non-terminals: \emph{concatenations} with production of the form $A \to BC$ (for symbols $B,C$)
and \emph{powers} with production of the form $A \to B^k$ (for a symbol $B$ and an integer $k\ge 2$). 
Every symbol $A$ generates a unique string denoted $\gen(A)$.

Each symbol $A$ is also associated with its \emph{parse tree} $\Tr(A)$ consisting of a root labeled with $A$ to which zero or more subtrees are attached:
if $A$ is a terminal, there are no subtrees;
if $A\to BC$ is a concatenation symbol, then $\Tr(B)$ and $\Tr(C)$ are attached;
if $A\to B^k$ is a power symbol, then $k$ copies of $\Tr(B)$ are attached.
Note that if we traverse the leaves of $\Tr(A)$ from left to right, spelling out the corresponding non-terminals,
then we obtain $\gen(A)$.
The parse tree $\Tr$ of the whole RSLP generating $T$ is defined as $\Tr(S)$ for the starting symbol $S$.
We define the \emph{value} $\val(v)$ of a node $v$ in $\Tr$ to be the fragment $T[a\dd b]$ corresponding to the leaves
$T[a],\ldots, T[b]$ in the subtree of $v$. 
Note that $\val(v)$ is an occurrence of $\gen(A)$, where $A$ is the label of $v$.
A sequence of nodes in $\Tr$ is a \emph{chain} if their values are consecutive fragments in $T$.

The \emph{recompression} technique by Je\.{z}~\cite{DBLP:journals/talg/Jez15,DBLP:journals/jacm/Jez16}
consists in the construction of a particular RSLP generating the input text $T$.
The underlying parse tree $\Tr$ is of depth $\Oh(\log n)$ and it can be constructed in $\Oh(n)$ time.
As observed by I~\cite{DBLP:conf/cpm/I17}, this parse tree $\Tr$ is \emph{locally consistent}
in a certain sense. To formalize this property, he introduced the \emph{popped sequence} of every fragment $T[a\dd b]$,
which is a sequence of symbols labelling a certain chain of nodes whose values constitute $T[a\dd b]$.

\begin{theorem}[\cite{DBLP:conf/cpm/I17}]\label{thm:recomp}
  If two fragments are equal, then their popped sequences are equal.
  Moreover, each popped sequence consists of $\Oh(\log n)$ runs (maximal powers of a single symbol) and can be constructed in $\Oh(\log n)$ time.
  The nodes corresponding to symbols in a run share a single parent.
  Furthermore, the popped sequence consists of a single symbol only for fragments of length $1$.
\end{theorem}

Let $F_1^{p_1}\cdots F_t^{p_t}$ be the run-length encoding of the popped sequence of a substring $S$ of $T$.
We define \[L(S) = 
  \{|\gen(F_1)|, |\gen(F_1^{p_1})|, |\gen(F_1^{p_1}F_2^{p_2})|, \ldots, |\gen(F_1^{p_1}\cdots F_{t-1}^{p_{t-1}})|,|\gen(F_1^{p_1}\cdots F_{t-1}^{p_{t-1}}F_{t}^{p_{t}-1})|\}.\]
By \cref{thm:recomp}, the set $L(S)$ can be constructed in $\Oh(\log n)$ time given the occurrence $T[a\dd b]=S$.

\begin{lemma}\label{lem:synchr}
  Let $v$ be a non-leaf node of $\Tr$ and let $T[a\dd b]$ be an occurrence of $S$ contained in $\val(v)$,
  but not contained in $\val(u)$ for any child $u$ of $v$.
  If $T[a\dd c]$ is the longest prefix of $T[a\dd b]$ contained in $\val(u)$ for a child $u$ of $v$,
  then $|T[a\dd c]| \in L(S)$.
  Symmetrically, if $T[c'+1\dd b]$ is the longest suffix of $T[a\dd b]$ contained in $\val(u)$ for a child $u$ of $v$,
  then $|T[a\dd c']| \in L(S)$.
\end{lemma}
\begin{proof}
Consider a sequence $v_1,\ldots,v_p$ of nodes in the chain corresponding to the popped sequence of $S=T[a\dd b]$.
Each of these nodes is a descendant of a child of $v$.
Note that  $T[a\dd c]=\val(v_1)\cdots \val(v_q)$, where $v_1,\ldots, v_q$ is the longest prefix consisting of descendants of the same child.
If the labels of $v_q$ and $v_{q+1}$ are distinct, then they belong to distinct runs
and $|T[a\dd c]|\in L(S)$.
Otherwise, $v_q$ and $v_{q+1}$ share the same parent: $v$.
Thus, $q=1$ and $|T[a\dd c]|=|\val(v_1)|\in L(S)$.
The proof of the second claim is symmetric.
\end{proof}

\paragraph{\bf Data Structure.}
We use recompression to build an RSLP generating $T$ and the underlying parse tree $\Tr$.
We also construct the component of \cref{lem:aux_break2} with $B_P = \{i + \frac12 : i \in L(P)\}$ for each pattern $P\in \D$.
Moreover, for every symbol $A$ we store the number of occurrences of patterns from $\D$ in $\gen(A)$.
Additionally, if $A \to B^k$ is a power, we also store the number of occurrences in $\gen(B^i)$ for $i\in [1\dd k]$.
The space consumption is $\Oh(n + d\log n)$ since $|B_P|=\Oh(\log n)$ for each $P\in \D$.

\paragraph{\bf Efficient preprocessing.}
The RSLP and the parse tree are built in $\Oh(n)$ time, and the sets $B_P$ are determined in $\Oh(d\log n)$ time using \cref{thm:recomp}.
The data structure of \cref{lem:aux_break2} is then constructed in $\Oh(n + d\log^{3/2} n)$ time.
Next, we process the RSLP in a bottom-up fashion. If $A$ is a terminal, its count is easily determined.
If $A\to BC$ is a concatenation, we sum the counts for $B$ and $C$ and the number of occurrences spanning both $\gen(B)$ and $\gen(C)$.
To determine the latter value, we fix an arbitrary node $v$ with label $A$ and denote its children  $u_\ell,u_r$.
By \cref{lem:synchr}, any occurrence of $P$ intersecting both $\val(u_\ell)$ and $\val(u_r)$ has a breakpoint aligned
to the inter-position between the two fragments.
Hence, the third summand is the result of a \textsc{Breakpoints-Anchor IDM} query
in $\val(v)$ with the anchor between $\val(u_\ell)$ and $\val(u_r)$.
Finally, if $A\to B^k$, then to determine the count in $\gen(B^i)$, we add the count for $B$, the count in $\gen(B^{i-1})$,
and the number of occurrences in $B^i$ spanning both the prefix $B$ and the suffix $B^{i-1}$.
To find the latter value, we fix an arbitrary node $v$ with label $A$, denote its children $u_1,\ldots, u_k$,
and make a \textsc{Breakpoints-Anchor IDM} query
in $\val(u_1)\cdots \val(u_i)$ with the anchor between $\val(u_1)$ and $\val(u_2)$.
The correctness of this step follows from \cref{lem:synchr}.
The running time of the last phase is $\Oh(n\log n / \log \log n)$, so the overall construction time
is $\Oh(n\log n / \log \log n + d\log^{3/2} n)$.

\paragraph{\bf Query.}
Upon a query $\C(i,j)$, we proceed essentially as in the warm-up solution: we recursively count the occurrences contained
in the intersection of $T[i\dd j]$ with $\val(v)$ for nodes $v$ in $\Tr$, starting from the root of $\Tr$.
If the two fragments are disjoint, the result is $0$, and if $\val(v)$ is contained in $T[i\dd j]$, it is the count precomputed for the label of $v$.
Otherwise, the label of $v$ is a non-terminal. If it is a concatenation symbol, we recurse on both children $u_\ell,u_r$ of $v$
and sum the obtained counts. If $T[i\dd j]$ spans both $\val(u_\ell)$ and $\val(u_r)$, we also
add the result of a \textsc{Breakpoints-Anchor IDM} query in the intersection of $T[i\dd j]$ with $\val(v)$ and the anchor between  $\val(u_\ell)$ and $\val(u_r)$. If the label is a power symbol $A\to B^k$, we determine which of the children $u_1,\ldots,u_k$ of $v$ are spanned by $T[i\dd j]$.
We denote these children by $u_\ell,\ldots, u_r$ and recurse on $u_\ell$ and on $u_r$.
If $r > \ell$, we also make a \textsc{Breakpoints-Anchor IDM} query in the intersection of $T[i\dd j]$ with $\val(u_\ell)\cdots \val(u_r)$ and anchor between $\val(u_\ell)$ and $\val(u_{\ell+1})$.
If $r>\ell+1$, we further add the precomputed value for $\gen(B^{r-\ell-1})$ to account for the occurrences contained in $\val(u_{\ell+1})\cdots \val(u_{r-1})$ and make a \textsc{Breakpoints-Anchor IDM} query in the intersection of $T[i\dd j]$ with $\val(u_{\ell+1})\cdots \val(u_r)$ and anchor between $u_{r-1}$ and $u_{r}$. By \cref{lem:synchr}, the answer is the sum of the up to five values computed.
The overall query time is $\Oh(\log^2 n / \log \log n)$, since we make $\cO(\log n)$ non-trivial recursive calls and each of them is processed in $\Oh(\log n / \log \log n)$ time.

\section{Dynamic dictionaries}\label{sec:dynamic}

In the Online Boolean Matrix-Vector Multiplication (OMv) problem, we are given as input an $n \times n$ boolean matrix $M$. Then, we are given in an online fashion a sequence of $n$ vectors $r_1, \ldots, r_n$, each of size $n$. For each such vector $r_i$, we are required to output $Mr_i$ before receiving $r_{i+1}$. 

\begin{conjecture}[OMv Conjecture~\cite{DBLP:conf/stoc/HenzingerKNS15}]\label{conj:OMv}
For any constant $\epsilon>0$, there is no $\cO(n^{3-\epsilon})$-time algorithm that solves OMv correctly with probability at least 2/3.
\end{conjecture}

We now present a restricted, but sufficient for our purposes, version of~\cite[Theorem 2.2]{DBLP:conf/stoc/HenzingerKNS15}.

\begin{theorem}[\cite{DBLP:conf/stoc/HenzingerKNS15}]\label{thm:gOMv}
\cref{conj:OMv} implies that there is no algorithm, for a fixed $\gamma>0$, that given as input an $r_1 \times r_2$ matrix $M$, with $r_1=\lfloor r_2^{\gamma} \rfloor$, preprocesses $M$ in time polynomial in $r_1+r_2$ and, then, presented with a vector $v$ of size $r_2$, computes $Mv$ in time $\cO(r_2^{1+\gamma-\epsilon})$ for $\epsilon > 0$, and has error probability of at most 1/3.
\end{theorem}

We proceed to obtain a conditional lower bound for IDM in the case of a dynamic dictionary.
This lower bound clearly carries over to the other problems we considered.

\begin{theorem}\label{thm:lowerbound}
The OMv conjecture implies that there is no algorithm that preprocesses $T$ and $\D$ in time polynomial in $n$, performs insertions to $\D$ in time $\cO(n^{\alpha})$, answers $\Dec(i,j)$ queries in time $\cO(n^{\beta})$, in an online manner, such that $\alpha+\beta =1-\epsilon$ for $\epsilon > 0$, and has error probability of at most 1/3.
\end{theorem}
\begin{proof}
Let us suppose that there is such an algorithm and set $\gamma = (\alpha + \epsilon/2) / (\beta+ \epsilon/2)$.
Given an $r_1 \times r_2$ matrix $M$, satisfying $r_1=\lfloor r_2^{\gamma} \rfloor$, we construct a text $T$ of length $n=r_1 r_2$ as follows. Let $T'$ be a text created by concatenating the rows of $M$ in increasing order. Then obtain $T$ by assigning to each non-zero element of $T'$ the column index of the matrix entry it originates from. Formally, for $i \in [1 \dd r_1r_2]$, let $a[i]=\lfloor (i-1)/r_2 \rfloor$ and $b[i]= i - a[i]r_2$ and set $T[i]=b[i]\cdot M[a[i]+1,b[i]]$.

We compute $Mv$ as follows. We add the indices of its at most $r_2$ non-zero entries in an initially empty dictionary. We then perform queries $\Dec(1+tr_2,(t+1)r_2)$ for $t=0,\ldots,r_1-1$. The answer to query $\Dec(1+tr_2,(t+1)r_2)$ is equal to the product of the $t$th row of $M$ with $v$. We thus obtain $Mv$. In total we perform $\cO(r_2)$ insertions to $\D$ and $\cO(r_1)$ $\Dec$ queries.
Thus, the total time required is $\cO(r_2 n^\alpha+r_1 n^\beta)=\cO(n^{\beta+\epsilon/2} n^\alpha + n^{\alpha+\epsilon/2} n^\beta)=\cO(n^{1-\epsilon/2})=\cO(r_2^{1+\gamma-\epsilon'})$ for $\epsilon' > 0$.
\cref{conj:OMv} would be disproved due to~\cref{thm:gOMv}.
\end{proof}

\begin{example}
  For the matrix \[M=\begin{bmatrix}
1 & 0 & 1 & 0 \\
0 & 0 & 1 & 1 \\
0 & 1 & 0 & 1
\end{bmatrix}\] we construct the text $T=1\,0\,3\,0\,0\,0\,3\,4\,0\,2\,0\,4$. 
For the vector $v=\begin{bmatrix}
1 & 1 & 0 & 0
\end{bmatrix}^T$, the dictionary is $\D=\{1,2\}$. The answers to $\Dec(1,4)$, $\Dec(5,8)$, $\Dec(9,12)$ are Yes, No, Yes, respectively, which corresponds to $Mv = \begin{bmatrix}
1 & 0 & 1
\end{bmatrix}^T$.
\end{example}

In the remainder of this section we focus on providing algorithms matching this lower bound.

We first summarize what is known for different kinds of internal pattern matching queries, in which at query time we are given a fragment  $T[i \dd j]$ and a substring $P$ of the text and ask queries analogous to those that we have defined for internal dictionary matching.
Note that we assume that the pattern $P$ is a substring of $T$ and is given by one of its occurrences. We answer $\C(P,i,j)$ queries as follows, similar to~\cite{DBLP:conf/spire/KopelowitzLP11}. We construct the suffix tree $\stt$ and preprocess it so that each node stores the lexicographic range of suffixes of which its path-label is a prefix.
We also construct a 2D orthogonal range counting data structure over an $n \times n$ grid $\mathcal{G}$, in which, for each suffix $T[a \dd n]$, we insert a point $(a,b)$, where $b$ is the lexicographic rank of this suffix among all suffixes.
We answer $\C(P,i,j)$ as follows.
We first locate the locus of $P$ in $\stt$ using a weighted ancestor query in $\cOtilde(1)$ time and retrieve the associated lexicographic range $[l,r]$. Next, we perform a counting query for the range $[i, j-|P|+1] \times [l,r]$ of $\mathcal{G}$, which returns the desired count; see also~\cite{MAKINEN2007332}.


\begin{theorem}[\cite{DBLP:journals/tcs/KellerKFL14,DBLP:conf/spire/KopelowitzLP11}]\label{thm:ipm}
$\Dec(P,i,j)$, $\R(P,i,j)$ and $\C(P,i,j)$ queries can be answered in time $\cOtilde(1+\occ)$ with an $\cOtilde(n)$-sized data structure that can be constructed in $\cOtilde(n)$ time.
\end{theorem}

Let us denote the dictionary we start with by $\D^0$.
Further, let $u_1, u_2, \ldots$ be the sequence of dictionary updates and $\D^r$ be the dictionary after update $u_r$.
Each update is an insertion or a deletion of a pattern in $\D$.
We first discuss how to answer $\RD$ queries.

\paragraph{\bf $\RD(i,j)$.}
We maintain the invariant that after update $u_t$ we have access to the static data structure of Section~\ref{sec:rd} for answering $\RD$ queries in $T$ with respect to dictionary $\D^r$, for some $r=t-\cO(m)$; note that $m$ will depend on $n+d$.
This can be achieved by rebuilding the data structure of Section~\ref{sec:rd} every $m$ updates in $\cOtilde(n+d)$ time, which amortizes to $\cOtilde((n+d)/m)$ time per update. (The time complexity can be made worst-case by application of the standard time-slicing technique.)
We also store updates $u_{r+1}, \ldots , u_t$ (or the differences between $\D^r$ and $\D^t$).

To answer a $\RD$ query, we do the following:
\begin{enumerate}[itemsep=0ex, parsep=1pt, topsep=1pt]
\item use the static data structure to answer the $\RD$ query for $\D^r$;
\item filter out the $\cO(m)$ reported patterns that are in $\D^r \setminus \D^t$;
\item search for the $\cO(m)$ patterns in $\D^t \setminus \D^r$ individually in $\cOtilde(1)$ time per pattern by performing internal pattern matching queries, employing~\cref{thm:ipm}.
\end{enumerate}
Each query thus requires time $\cOtilde(m+\occ)$. We arrive at the following proposition.

\begin{proposition}\label{lem:rddyn}
  The $\RD(i,j)$ queries for a dynamic dictionary can be answered in $\cOtilde(m+\occ)$ time  per query and $\cOtilde((n+d)/m)$ time per update for any ${m\in [1\dd n+d]}$  using $\cOtilde(n+d)$ space.
\end{proposition}

We next show how to attain $\cOtilde(n^{\alpha})$ update time and $\cOtilde(n^{1-\alpha}+\occ)$ query time for any $0<\alpha <1$. In other words, we show how to avoid the direct dependency on the size of the dictionary.

We store $\D$ as an array $D$ of collections so that a pattern $P \in \D$ is stored in $D[p]$, where $p$ is the starting position of the leftmost occurrence of $P$ in $T$. We can find the desired position $p$ for a pattern $P$ in $\cO(\log\log n)$ time by locating its locus on $\stt$ using a weighted ancestor query; we can have precomputed the leftmost occurrence of the path-label of each explicit node in a DFS traversal of $\stt$. Let each collection store its elements in a min heap with respect to their lengths.

The dictionary $\D'=\{\min D[p] : 1 \leq p \leq n\}$ is of size $\cO(n)$.
An insertion in $\D$ corresponds to a possible insertion followed by a possible deletion in $\D'$, while a deletion in $\D$ corresponds to a possible deletion in $\D'$, followed by an insertion if the collection in $D$ where the deletion occurs is non-empty.
We observe that if some $P \in \D \setminus \D'$ occurs in $T[i \dd j]$, then the minimum element in the collection in which $P$ belongs also occurs in $T[i \dd j]$.

We thus use the solution for $\RD(i,j)$ from~\cref{lem:rddyn} for $\D'$. We then iterate over the elements of each collection from which an element has been reported in the order of increasing length, while they occur in $T[i \dd j]$; we check whether this is the case using~\cref{thm:ipm}.

\paragraph{\bf $\R(i,j)$.}
We first perform a $\RD(i,j)$ query and then find all occurrences of each returned pattern in $T[i \dd j]$ in time $\cOtilde(1+\occ)$ by~\cref{thm:ipm}.

\paragraph{\bf $\Dec(i,j)$.}
We again use $\D'$.
We first use the static version of $\C(i,j)$, presented in~\cref{sec:count} and then the counting version of internal pattern matching for removed/added patterns using~\cref{thm:ipm}, incrementing/decrementing the counter appropriately. We rebuild the data structure every $m$ updates.

\paragraph{\bf $\C(i,j)$.}
We first build the data structure of Section~\ref{sec:count} for $\C(i,j)$ queries for dictionary  $\D^0$. For the subsequent $m$ updates, we answer $\C(i,j)$ queries, using this data structure and treating individually the added/removed patterns using Theorem~\ref{thm:ipm}.
These queries are thus answered in $\cOtilde(m)$ time.

After every $m$ updates, we update our data structure to refer to the current dictionary as follows. (We focus on $\D^0$ and  $\D^m$ for notational simplicity.)
We update the counts of occurrences for all nodes of $\Tr$ by computing the counts for the set of added and the set of removed patterns in $\cOtilde(n)$ time and updating the previously stored counts accordingly.

As for \textsc{Breakpoints-Anchor IDM}, we also have to do something smarter than simply recompute the whole data structure from scratch, as we do not want to spend $\Omega(d)$ time.
At preprocessing, we set our grid $\mathcal{G}$ to be of size $K \times K$ for $K=\cO(n^2)$ and identify $x$-coordinate $i$ with the $i$th smallest element of the set $W=\{U x : U \text{ a substring of } T\text{ and }x \in \{\$,\#\}\}$. (Similarly for $y$-coordinates and $T^R$.)

We can preprocess the suffix tree $\stt$ in $\cO(n)$ time so that the rank of a given $T[a \dd b] \$$ or $T[a \dd b]\#$ in $W$ can be computed in $\cOtilde(1)$ time. 
Let us assume that $\stt$ has been built for $T$, without $\$$ appended to it.
We make a DFS traversal of $\stt$, maintaining a global counter $\textsf{cr}$, which is initialized to zero at the root. The DFS visits the children of a node in a left-to-right order. When traversing an edge, we increment $\textsf{cr}$ by the size of the path-label of this edge. When an explicit node $v$ is visited for the \emph{first} time we set the rank of $\Lab(v)\$$ equal to $\textsf{cr}$; if $v$ is a leaf then $\textsf{cr}$ is incremented by one.
The rank of $\Lab(v)\#$ is set to $\textsf{cr}$ when $v$ is visited for the \emph{last} time.
Let $q$ be the locus of $T[a \dd b]$ in $\stt$, which can be computed in $\Oh(\log \log n)$ time using a weighted ancestor query.
If $q$ is an explicit node, the ranks of $T[a\dd b]\$$ and $T[a\dd b]\#$ are already stored at $q$. 
Otherwise, these ranks can be inferred from the ranks of $\Lab(v)\$$ and $\Lab(v)\#$ stored at the nearest explicit descendant $v$ of $q$
by, respectively, subtracting and adding the distance between $v$ and $q$.

Thus, instead of explicitly building trees $W$ and $W^R$ as in the proof of Lemma~\ref{lem:aux_break1}, we use $\stt$ and $\mathcal{T}(T^R)$ and maintain rectangles in $\mathcal{G}$.
After $m$ updates, we remove (resp.~add) the $\cOtilde(m)$ rectangles corresponding to patterns in $\D^0 \setminus D^m$ (resp.~$\D^m \setminus D^0$).
We can maintain a data structure of size $\cO(r) $for the counting version of range stabbing in a 2D grid of size $K\times K$ with $r$ rectangles with $\cO(\log K \log r/ (\log\log r)^2)$ time per update and query~\cite{DBLP:journals/comgeo/HeM14}.

To wrap up, updating the data structure every $m$ updates requires $\cOtilde(1+n/m)$ amortized time. We can deamortize the time complexities using the time-slicing technique. 
This concludes the proof of the following theorem.

\begin{theorem}\label{thm:rddyn}
$\Dec(i,j)$, $\R(i,j)$, $\RD(i,j)$, and $\C(i,j)$ queries for a dynamic dictionary can be answered in $\cOtilde(m+\occ)$ time  per query and $\cOtilde(n/m)$ time per update for any parameter $m\in [1\dd n]$ using $\cOtilde(n+d)$ space.
\end{theorem}

\section{Final Remarks}
The question of whether answering queries of the type $\CD(i,j)$, i.e.~returning the number $c$ of patterns from $\D$ that occur in $T[i \dd j]$, can be answered in time $o(\min\{c,|j-i|\})$ or even $\cOtilde(1)$ with a data structure of size $\cOtilde(n+d)$ is left open for further investigation.
It turns out that our techniques can be used to efficiently answer such queries $\cO(\log n)$-approximately; details can be found in Appendix~\ref{sec:CD}.

\paragraph{\bf Acknowledgements.}
Pangiotis Charalampopoulos and Manal Mohamed thank Solon Pissis for preliminary discussions.

\bibliographystyle{plainurl}
\bibliography{idm}

\appendix
\section{An approximation for $\CD(i,j)$}\label{sec:CD}
In this section we consider queries of the type $\CD(i,j)$, i.e.~returning the number of patterns from $\D$ that occur in $T[i \dd j]$.
It turns out that our techniques can be used to efficiently answer such queries $\cO(\log n)$-approximately.

Let us first note that the hardness of the considered problem partially stems from the fact that it is not decomposable.
In this section we adapt the data structure presented in Section~\ref{sec:count} for answering $\C(i,j)$ queries and show the following result.

\begin{theorem}\label{thm:cd}
Queries of the form $\CD(i,j)$ can be answered $\cO(\log n)$-approximately in $\cO(\log^2 n/ \log\log n)$ time with a data structure of size $\cO(n +d \log^2 n)$ that can be constructed in $\cO(n \log n +d\log^{5/2} n)$ time.
\end{theorem}

\subsection{An $\cO(\log^2 n)$-approximation}
We modify the data structure of Section~\ref{sec:count} for $\C(i,j)$ as follows. For each symbol $A$ in $\Tr$, we store the number of patterns from $\D$ that occur in $\gen(A)$. Additionally, if $A \to B^k$ is a power, we also store the number of patterns from $\D$ that occur in $\gen(B^i)$ for $i\in [1\dd k]$.

At query, we imitate the algorithm for $\C(i,j)$. We count each distinct pattern that occurs in $T[i \dd j]$ at least once and at most $\cO(\log^2 n)$ times---at most once per visited node $v$ such that $\val(v)$ is contained in $T[i \dd j]$ and at most $\cO(\log n)$ times for each considered anchor. The latter follows from the fact that the grid $\mathcal{G}$ that is used to answer \textsc{Breakpoint-Anchor IDM} has $\cO(\log n)$ rectangles for each pattern.

\paragraph{\bf Efficient preprocessing.}
We first show how to compute the desired counts for all nodes of the parse tree in time $\cO(d+\sum_{v \in \Tr} |\val(v)|)=\cO(d+n \log n)$ in total.
For a symbol $A$, with $\gen(A)=T[a \dd b]$, for each $i \in [a \dd b]$, we compute the locus of the longest prefix of $T[i \dd b]$ that is in $\D$ in the $\D$-modified suffix tree of $T$ using a weighted ancestor query. (We answer all weighted ancestor queries for all nodes as a batch, similar to before, in the stated time complexity.)
Let the respective loci be $u_1,u_2, \ldots, u_{b-a+1}$, sorted in the lexicographical order of their path-labels.
We process the loci in this order.
We process $u_i$ by incrementing the count by the number of nodes on the path $(\LCA(u_{i-1},u_i),u_i]$, where $u_0$ is the root of the tree.\footnote{A lowest common ancestor (\LCA) query takes as input two nodes of a rooted tree and returns the deepest node of the tree that is an ancestor of both. Such queries can be answered in constant time after a linear-time preprocessing of the tree~\cite{DBLP:journals/jal/BenderFPSS05}.}
(We assume that we have precomputed the number of ancestors of each node of the tree.)
The correctness of this algorithm follows from a straightforward inductive argument and the fact that the deepest node among $\LCA(u_j,u_i)$, for $j<i$, is $\LCA(u_{i-1},u_i)$.

Finally, we argue that the time required to compute the number of patterns from $\D$ that occur in $\gen(B^i)$ for $i\in [1\dd k]$ for a symbol $A \to B^k$ is also $\cO(d+n \log n)$.
This follows from the fact that any pattern occurring in $\gen(B^i)$ has an occurrence in the first $|\gen(B)|$ positions. Hence, it suffices to compute the required loci only for these for each $i$.
We thus consider $k \cdot |\gen(B)| = |\gen(A)|$ such loci in total for this step of the computation as well.

\subsection{Refinement through Colored Range Stabbing Counting}

We improve the approximation ratio as follows. We color all rectangles corresponding to the same pattern with the same color in our data structure for \textsc{Breakpoint-Anchor IDM} and upon query ask a color rectangle stabbing query, i.e.~ask for the number of distinct colors that the rectangles stabbed by the point have.

\begin{lemma}\label{lem:4sided}
Given $k$ 4-sided colored rectangles in 2D, such that the number of rectangles colored with each color are $\cO(c)$, we can answer colored range stabbing counting queries in $\cO(\log k / \log \log k)$ time with a data structure of size $\cO(c \cdot k)$ that can be constructed in $\cO(c \cdot k \sqrt{\log k})$ time.
\end{lemma}
\begin{proof}
We first make the following claim.

\begin{claim}
Given a collection $\mathcal{R}$ of $\cO(c)$ axes-parallel rectangles in 2D, we can construct a collection $\mathcal{R}'$ of $\cO(c^2)$ non-overlapping rectangles in 2D in $\cO(c^2)$ time such that for every $(a,b) \in \Z^2$, $(a,b)$ stabs some rectangle in $\mathcal{R}$ if and only if it stabs some rectangle in $\mathcal{R}'$.
\end{claim}
\begin{proof}
We run a textbook line-sweeping algorithm, efficiently maintaining information about points contained in some rectangle using a segment tree.
\end{proof}

We apply the above claim to the rectangles of each color.
Let us now note that a colored stabbing counting query for our original set of rectangles corresponds to a simple stabbing counting query for the new set and we can thus employ Theorem~\ref{thm:range_count}.
\end{proof}

Recall that there are $\cO(d \log n)$ rectangles in total, $\cO(\log n)$ with each color. We plug in the data structure of Lemma~\ref{lem:4sided} in our solution for \textsc{Breakpoint-Anchor IDM} and obtain Theorem~\ref{thm:cd}.

\end{document}